\documentclass{amsart}
\usepackage{amsmath,amsfonts,amsbsy,amsgen,amscd,mathrsfs,amssymb,amsthm,mathtools}
\usepackage[colorlinks=true,citecolor=blue,linkcolor=blue]{hyperref}

\usepackage{accents}

\usepackage{geometry}
\numberwithin{equation}{section}
\newtheorem{theorem}{Theorem}[section]
\newtheorem{corollary}[theorem]{Corollary}

\newtheorem{lemma}[theorem]{Lemma}
\newtheorem{proposition}[theorem]{Proposition}
\theoremstyle{definition}
\newtheorem{assumption}[theorem]{Assumption}
\newtheorem{definition}[theorem]{Definition}

\newtheorem{notation}[theorem]{Notation}
\newtheorem{problem}[theorem]{Problem}

\newtheorem{remark}[theorem]{Remark}

\makeatletter\renewenvironment{proof}[1][\proofname] {\par\pushQED{\qed}\normalfont\topsep6\p@\@plus6\p@\relax\trivlist\item[\hskip\labelsep\bfseries#1\@addpunct{.}]\ignorespaces}{\popQED\endtrivlist}

\newcommand\al{\alpha}
\newcommand\be{\beta}
\newcommand\dd{\mathrm d}

\newcommand\de{\delta}

\newcommand\eps{\varepsilon}
\newcommand\Ga{\Gamma}
\newcommand\ga{\gamma}
\newcommand\ka{\kappa}
\newcommand\La{\Lambda}
\newcommand\la{\lambda}
\newcommand\Om{\Omega}
\newcommand\om{\omega}

\newcommand\si{\sigma}

\newcommand\ze{\zeta}
\renewcommand\d{~\mathrm d}

\renewcommand\phi{\varphi}
\renewcommand\rho{\varrho}

\newcommand\mbb{\mathbb}
\newcommand\mbf{\mathbf}
\newcommand\mc{\mathcal}
\newcommand\mf{\mathfrak}
\newcommand\mr{\mathrm}
\newcommand\ms{\mathscr}
\newcommand\msf{\mathsf}

\begin{document}

\title[Spectral Conditioning of Discrete Random Schr\"odinger Operators]
{On Spatial Conditioning of the Spectrum of Discrete Random Schr\"odinger Operators}
\author{Pierre Yves Gaudreau Lamarre}
\address{University of Chicago, Chicago, IL\space\space60637, USA}
\email{pyjgl@uchicago.edu}
\author{Promit Ghosal}
\address{Massachusetts Institute of Technology, Cambridge, MA\space\space02139, USA}
\email{promit@mit.edu}
\author{Yuchen Liao}
\address{University of Warwick, Coventry\space\space CV4 7AL, UK}
\email{Yuchen.Liao@warwick.ac.uk}
\maketitle

\begin{abstract}
Consider a random Schr\"odinger-type operator of the form $H:=-H_X+V+\xi$ acting on a general
graph $\ms G=(\ms V,\ms E)$, where $H_X$ is the generator of a Markov process $X$ on $\ms G$, $V$ is a deterministic
potential with sufficient growth (so that $H$ has a purely
discrete spectrum), and $\xi$ is a random noise with at-most-exponential tails.
We prove that
the eigenvalue point process of $H$ is number rigid in the sense of Ghosh and Peres \cite{GP17}; that is,
the number of eigenvalues in any bounded domain $B\subset\mbb C$ is determined by the configuration
of eigenvalues outside of $B$.
Our general setting allows to treat cases where $X$ could be non-symmetric (hence $H$ is non-self-adjoint) and $\xi$ has long-range dependence.
Our strategy of proof consists of controlling the variance of the
trace of the semigroup $\mr e^{-t H}$ using the Feynman-Kac formula.
\end{abstract}

\section{Introduction}

Let $\ms G=(\ms V,\ms E)$ be a countably infinite connected graph
with uniformly bounded degrees and a distinguished vertex $0\in\ms V$, which we call the root.
For example, $\ms G$ could be the integer lattice $\mbb Z^d$,
any semiregular tessellation/honeycomb of $\mbb R^d$ that includes the
origin, or a much more general graph.

In this paper, we are interested in the spectral theory of
random Schr\"odinger-type operators of the form
\[Hf(v)=-H_Xf(v)+\big(V(v)+\xi(v)\big)f(v),\qquad v\in\ms V,~f:\ms V\to\mbb R,\]
where we assume that
\begin{enumerate}
\item $H_X$ is the infinitesimal generator of some continuous-time
Markov process $X$ on $\ms G$ (which need not be symmetric);
\item $\xi:\ms V\to\mbb R$ is a random noise (which may have long-range
dependence); and
\item $V:\ms V\to\mbb R\cup\{\infty\}$ is a deterministic potential with sufficient growth
at infinity (as measured by the size of $V(v)$ as $v$ grows farther away from
the root), ensuring that $H$ has a purely discrete spectrum.
\end{enumerate}
More specifically, we are interested in studying the {\it spatial conditioning} of the spectrum
of $H$, i.e., understanding the random configuration of $H$'s eigenvalues in some domain
$B\subset\mbb C$ conditional on the configuration of eigenvalues outside of $B$.
As a first step in this direction, we establish that under general assumptions on $H_X$, $\xi$, and $V$,
$H$'s spectrum is {\it number rigid} in the sense of Ghosh and Peres \cite{GP17}; that is,
the number of eigenvalues of $H$ in bounded domains $B\subset\mbb C$
is a measurable function of the configuration of $H$'s eigenvalues outside of $B$
(we point to Definition \ref{Definition: Rigidity} for a precise definition).      
To the best of our knowledge, ours is the first work to study the occurrence of
such a phenomenon in the spectrum of random Schr\"odinger operators acting on discrete spaces.

The spectral theory of differential operators (including non-self-adjoint operators; e.g.,
\cite{AAD01,Bo17,davies1980,DN02,DHK09,FKV18,H17,KL18,LS09}) is among the
most promiment research programs in mathematical physics; see, for instance, 
\cite{HislopSigal,Teschl}. In particular, starting from the pioneering
work of Anderson \cite{Anderson58}, the study of Sch\"odinger
operators perturbed by irregular noise has attracted a lot of attention; we refer
to \cite{AizenmanWarzel,CarLac} for general introductions to the subject. A particularly active program in this direction
is the work on {\it Anderson localization}, which concerns the appearance of
pure point spectrum and eigenfunction decay; see the survey articles \cite{H08,K08,St11}
for more details.

In contrast to localization and similar questions, in this paper we investigate the {\it transport
of spectral information} from one region to another,
whereby observing the configuration of $H$'s eigenvalues in some domain $D\subset\mbb C$
allows to recover nontrivial information about the spectrum in $D$'s complement.
Such questions of spatial conditioning in general point processes
have long been of interest due to their
natural applications in mathematics and physics; see, e.g., \cite{ApplicationsPtProcess,Ka17}.
In recent years, there has been a renewed interest in such investigations coming from
the seminal work of Ghosh and Peres \cite{GP17} on {\it rigidity and tolerance},
culminating in a now active field of research (e.g.,
\cite{Buf16,Buf18,BDQ18,BNQ18,G15,SG16,GhoshKrishnapur,Ghosh17,PeresSly};
see also \cite{AM80}).
In \cite{GGL20}, we studied the occurrence of number rigidity in the
spectrum of a class of random Schr\"odinger operators on one-dimensional continuous space.
In this paper, we study a similar problem for discrete random Schr\"odinger operators.

\subsection{Organization}
In the remainder of this introduction, we provide an outline of our main results
and proof strategy, we compare the results in this paper to previous investigations
in a similar vein, and we discuss a few natural open questions raised by our work.

In Section \ref{Section: Outline}, we provide a high-level
outline of the proof of our main results. We take this opportunity to
explain how our technical assumptions arise from our computations.
In Section \ref{sec:Main Result}, we state our assumptions
and main results in full details, namely,
Assumptions \ref{Assumption: Graph} and \ref{Assumption: Potential and Noise}
and Theorems \ref{Theorem: Upper}, \ref{Theorem: Rigidity}, and \ref{Theorem: Lower}.
Then, we prove Theorem \ref{Theorem: Upper} in Section \ref{sec: Proof of Upper},
we prove Theorem \ref{Theorem: Rigidity} in Sections \ref{sec: Multiplicity} and \ref{Section: Rigidity},
and we prove Theorem \ref{Theorem: Lower} in Section \ref{sec: Proof of Lower}.

\subsection{Outline of Main Results}

Let $\msf d$ denote the graph distance on $\ms G$.
For every $v\in\ms V$, we use $\msf c_n(v)$, $n\geq0$, to denote $v$'s coordination sequence
in $\ms G$; that is, for every $n\in\mbb N$,
$\msf c_n(v)$ is the number of vertices $u\in\ms V$ such that $\msf d(u,v)=n$. Stated informally,
our main result is as follows:

\begin{theorem}[Informal Statement]\label{thm:Informal}
Suppose that there exists $d\geq1$ such that
\begin{align}
\label{Equation: Coordination 1}
\sup_{v\in\ms V}\msf c_n(v)=O(n^{d-1})\qquad\text{as }n\to\infty.
\end{align}
Under mild technical assumptions on the Markov process $X$ and
the noise $\xi$,
there exists a constant $d/2\leq\al\leq d$ (which, apart from $d$, depends on the
the range of the covariance in $\xi$) such that if $V(v)$ grows
faster than $\msf d(0,v)^\al$ as $\msf d(0,v)\to\infty$,
then
the eigenvalue point process of $H$ is number rigid.
\end{theorem}

See Theorems \ref{Theorem: Upper}
and \ref{Theorem: Rigidity} for a formal statement.
Our technical assumptions are stated in Assumptions \ref{Assumption: Graph}
and \ref{Assumption: Potential and Noise}; roughly speaking,
our assumptions are that
\begin{enumerate}
\item the jump rates of $X$ (which may be site-dependent) are uniformly bounded; and
\item the tails of $\xi$ are not worse than exponential.
\end{enumerate}
In particular, our assumptions allow for $X$ to be non-symmetric
(hence, the operator $H$ need not be self-adjoint) and for $\xi$
to have a variety of covariance structures, including long-range
dependence.

\begin{remark}
\label{Remark: Dimension}
The constant $d$ in \eqref{Equation: Coordination 1}, which quantifies the growth rate of the number of vertices,
can be thought of as the {\it dimension} of $\ms G$ (or, at least, an upper bound of the dimension).
To illustrate this, if $\ms G$ is for example $\mbb Z^d$ or a semiregular tessellation of $\mbb R^d$,
then it is easy to see that $cn^{d-1}\leq \msf c_n(v)\leq Cn^{d-1}$ for some $C,c>0$. More generally,
the constant $d$ is closely related to the {\it intrinsic dimension} of $\ms G$, which is the minimal number
$k$ such that $\ms G$ can be embedded in $\mbb Z^k$. We refer to, e.g., \cite{KL07,LLR95} for more details.
\end{remark}

\begin{remark}
In Theorem \ref{Theorem: Lower}, we provide concrete examples showing
that the growth lower bound of $\msf d(0,v)^\al$ that we impose on $V$ to
get rigidity is the best general sufficient condition that can be obtained with our proof method.
The question of whether or not this is actually necessary for rigidity is addressed in
Section \ref{sec:new methods}.
\end{remark}

\subsection{Proof Strategy and Previous Results}

Our method to prove number rigidity follows the general scheme introduced by Ghosh and Peres
in \cite{GP17}: Let $\mc X=\sum_{k\in\mbb N}\de_{\la_k}$ be a point process on $\mbb C$.
As per \cite[Theorem 6.1]{GP17}, for any bounded set $B\subset\mbb C$, if there exists
a sequence of functions $(f_n)_{n\in\mbb N}$ such that, as $n\to\infty$,
\begin{enumerate}
\item $f_n\to1$ uniformly on $B$, and
\item the variance of the linear statistics
$\int f_n\d\mc X=\sum_{k\in\mbb N}f_n(\la_k)$ vanish,
\end{enumerate}
then $\mc X(B)$ is measurable
with respect to the configuration of $\mc X$ outside of $B$.

One of the main difficulties involved with carrying out the above program lies in
the computation of upper bounds for the variances of linear statistics $\mbf{Var}[\int f\d\mc X]$.
For this reason, much of the previous literature on number rigidity exploits special
properties that make the computations more manageable, such as
determinantal/Pfaffian or other inetegrable structure \cite{Buf16,BNQ18,G15,GL18,GP17}, translation
invariance and hyperuniformity \cite{GS19,Ghosh17}, and finite dimensional approximations \cite{RN18}.

Among those works, the only result that is related to the spectrum of random Schr\"odinger operators
is the proof of rigidity of the Airy-2 point process in \cite{Buf16}. Thanks to the work of Edelman,
Ram{\'\i}rez, Rider, Sutton, and Vir\'ag \cite{EdelmanSutton,RamirezRiderVirag}, this implies that the spectrum
of the {\it stochastic Airy operator} with parameter $\be=2$ is number rigid. Given that the
method of proof in \cite{Buf16} relies crucially on special algebraic structure only present in that
one particular case, however, the result cannot be extended to general Schr\"odinger operators.

More recently, in \cite{GGL20} we proposed to study number rigidity in
the spectrum of random Schr\"odinger operators using a new {\it semigroup method}: Given that
the exponential functions $\mr e_n(z):=\mr e^{-z/n}$ converge uniformly to 1 on any bounded
set as $n\to\infty$, in order to prove number rigidity of any point process, it suffices to prove
that $\mbf{Var}[\int \mr e_n\d\mc X]\to0$ (though the requirement that $\int \mr e_n\d\mc X$
is finite imposes strong conditions on $\mc X$).
If $\mc X$ happens to be
the eigenvalue point process of a random Schr\"odinger operator $H$, then $\int \mr e_n\d\mc X$
is the trace of the operator $\mr e^{-H/n}$. Thus, in order to prove the number rigidity of the
spectrum of any random Schr\"odinger operator $H$, it suffices to prove that
\[\lim_{t\to0}\mbf{Var}\big[\mr{Tr}[\mr e^{-t H}]\big]=0.\]
The reason why this is a particularly attractive strategy to prove number rigidity of general
random Schr\"odinger operators is that, thanks to the Feynman-Kac formula, there
exists an explicit probabilistic representation of the semigroup $(\mr e^{-t H})_{t>0}$ in terms
of elementary stochastic processes, making the variance $\mbf{Var}\big[\mr{Tr}[\mr e^{-t H}]\big]$
amenable to computation.

In \cite{GGL20}, this strategy was used to prove number rigidity
for a class of random Schr\"odinger operators acting on one-dimensional continuous
space (i.e., an interval of the form $I=(a,b)$ with $-\infty\leq a<b\leq\infty$). In this paper,
we apply the same methodology to prove number rigidity for a general class of discrete
random Schr\"odinger operators.

Despite the fact that the general strategy of proof used in the present paper is the same as
\cite{GGL20}, the differences between the two settings are such that
virtually none of the work carried out in \cite{GGL20} can be directly
extended to the present paper. For example:
\begin{enumerate}
\item Since we consider operators acting on general graphs $\ms G$,
the treatment of the geometry of the space on which our operators are
defined requires a much more careful analysis than that carried out in
\cite{GGL20}. In particular (as per Remark
\ref{Remark: Dimension}), in this paper we uncover
that the dimension of the space plays an important role in the proof
of rigidity using the semigroup method.
\item In \cite{GGL20}, we only consider Schr\"odinger
operators whose kinetic energy operator is the standard Laplacian and whose noise
is a Gaussian process. As a result, the operators considered therein are all self-adjoint
and upper bounds of $\mbf{Var}\big[\mr{Tr}[\mr e^{-t H}]\big]$ can mostly be reduced to the analysis
of self-intersection local times of standard Brownian motion.
In contrast, in this paper we allow for much more general generators $H_X$ and noises $\xi$.
Most notably, the assumptions of this paper allow for non-self-adjoint operators, which
increases the technical difficulties involved (e.g., Sections \ref{sec: Multiplicity}
and \ref{Section: Rigidity}).
\end{enumerate}

\subsection{Future Directions}

Given that our main theorems apply to a very general class of operators,
the results of this paper provide substantial evidence of the universality
of number rigidity in discrete random Schr\"odinger operators.
That being said, we feel that our results raise a number of interesting
follow-up questions. We now discuss three such directions.

\subsubsection{New Methods}
\label{sec:new methods}

It is natural to wonder if the growth condition $V(v)\gg\msf d(0,v)^\al$
that we impose on the potential to get number rigidity is close to optimal. As we
show in Theorem \ref{Theorem: Lower}, our main result is optimal in the sense that
we can find concrete examples of operators such that
\begin{align}
\label{Equation: Non-Vanishing of Exponential}
\liminf_{t\to0}\mbf{Var}\big[\mr{Tr}[\mr e^{-t H}]\big]>0
\end{align}
when $V(v)\asymp\msf d(0,v)^\al$. That being said, the vanishing of the variance of the
trace of the semigroup is only a sufficient condition for number rigidity, and, in fact,
it was observed in \cite[Proposition 2.27]{GGL20} that there
exists at least one random Schr\"odinger operator whose spectrum is known to
be number rigid and such that \eqref{Equation: Non-Vanishing of Exponential}
holds. For example, the following simple question appears to be outside the scope
of the methods used in this paper:

\begin{problem}
\label{Problem: Rigidity vs. Dimension}
Suppose that $X$ is the simple symmetric random walk on $\ms G=\mbb Z^d$,
that $V(v)=\msf d(0,v)^\de$ for some $\de>0$, and that $\big(\xi(v)\big)_{v\in\mbb Z^d}$
are i.i.d. standard Gaussians (or any other simple distribution). Is $H$'s spectrum always number rigid in this case?
\end{problem}

More specifically, given that $\msf c_n(v)\asymp n^{d-1}$ on $\mbb Z^d$,
our main theorem only implies number rigidity in the above when $\de>d/2$.
We expect that solving Problem \ref{Problem: Rigidity vs. Dimension}
will require developing new methods to study number rigidity in random
Schr\"odinger operators.

\subsubsection{The Mechanism of Rigidity}
\label{Section: Mechanism}

Our main result implies that for every bounded measurable set $B\subset\mbb C$,
there exists a deterministic function $\mc N_B$ such that
the identity
\[\text{number of $H$'s eigenvalues in $B$}=\mc N_B(\text{configuration of $H$'s eigenvalues outside $B$})\]
holds with probability one. That being said, the argument that we use to prove the
existence of $\mc N_B$
gives little information on its exact form. In other words, the precise nature
of the mechanism that makes the number of eigenvalues in $B$ a deterministic function of the configuration
on the outside remains largely unknown. In light of this, an interesting future direction for investigation would be
along the following lines:

\begin{problem}
Let $B\subset\mbb C$ be a ``simple" bounded subset of the complex plane
(e.g., a closed or open ball).
Does $\mc N_B$ admit an explicit representation?
\end{problem}

We point to Remark \ref{Remark: Mechanism} for more details on the construction of $\mc N_B$.

\subsubsection{Spatial Conditioning Beyond Number Rigidity}

When $H$'s spectrum is number rigid, we know that if we condition
$H$ on having a specific eigenvalue configuration outside of a bounded set
$B$, then $H$'s spectrum inside of $B$ is a point process with a fixed total number
of points. It would be interesting to see if more can be learned about
the conditional distribution of the eigenvalues in $B$. For instance, the
following problem (related to the notion of tolerance introduced in
\cite{GP17}) might be a good starting point:

\begin{problem}
Suppose that, after conditioning on the outside configuration,
$H$ has $M\in\mbb N$ random eigenvalues in some bounded
set $B\subset\mbb C$. Let $\La\in\mbb C^M$ be the random
vector whose components are the random eigenvalues of $H$ in $B$
(conditional on the configuration outside $B$),
taken in a uniformly random order. What is the support of $\La$'s
probability distribution on the set $B^M$?
\end{problem}

\section{Proof Outline}
\label{Section: Outline}

In this section, we present a sketch of the proof of
our main theorem in two simple special cases.
We take this opportunity to explain how our technical
assumptions arise in our computations.
For simplicity of exposition, we assume in this outline that $\ms G$ is the integer lattice
$\mbb Z^d$ (i.e., $(u,v)\in\ms E$ if and only if $\|u-v\|_\infty=1$, where $\|\cdot\|_\infty$ denotes the
usual $\ell^\infty$ norm), $X$ is the simple
symmetric random walk on $\mbb Z^d$, and $\xi$ is a centered stationary Gaussian process
with covariance function
\[\ga(v):=\mbf E[\xi(v)\xi(0)],\qquad v\in\mbb Z^d.\]

As alluded to in the introduction (and proved in Section \ref{Section: Rigidity}),
to prove that the eigenvalue point process of $H$ is number rigid, it suffices to show that
$\mr{Tr}[\mr e^{-t H}]$'s variance vanishes as $t\to0$.
According to the Feynman-Kac formula, we have that
\[\mr{Tr}[\mr e^{-t H}]=\sum_{v\in\mbb Z^d}\mbf E_X\left[\exp\left(\int_0^t V\big(X(s)\big)+\xi\big(X(s)\big)\d s\right)\mbf 1_{\{X(t)=X(0)\}}\bigg|X(0)=v\right],\]
where $\mbf E_X$ means that we are only averaging with respect to the randomness in the path of $X$,
and we assume that $X$ is independent of the noise $\xi$.
In order to ensure that $\mr e^{-t H}$ is trace class (or even bounded) in the general case,
we assume that $\ms G$ has uniformly bounded degrees; see Section \ref{Section: Boundedness}
for more details.

Our first step in the analysis of $\mr{Tr}[\mr e^{-t H}]$ is to note that if
$t$ is small, then the probability that there exists some $0\leq s\leq t$ such that $X(s)\neq X(0)$
is close to zero (i.e., $1-\mr e^{-t}\sim t$). Thus, by working only with the complement of this event,
we have that
\begin{align}
\label{Equation: Heuristic 1}
\mr{Tr}[\mr e^{-t H}]\approx\sum_{v\in\mbb Z^d}\mr e^{-tV(v)-t\xi(v)}.
\end{align}
A rigorous version of this heuristic is carried out in
the proof of Lemma \ref{Lemma: Variance Upper Bound 3}.
The latter relies on controlling how far $X$ can travel from its initial value $X(0)$ after a small time
(e.g., the tail  bound \eqref{Equation: Tail Bound}), which itself depends on the
assumptions that the jump rates of $X$ are uniformly bounded.

Our second step is to identify the leading order asymptotics in the variance
of the expression on the right-hand side of \eqref{Equation: Heuristic 1}.
In the special case where $\xi$ is a stationary Gaussian process with covariance $\ga$,
an application of Tonelli's theorem yields
\begin{align}
\nonumber
\mbf{Var}\left[\sum_{v\in\mbb Z^d}\mr e^{-tV(v)-t\xi(v)}\right]
&=\sum_{u,v\in\mbb Z^d}\mr e^{-tV(u)-tV(v)}\mbf{Cov}[\mr e^{-t\xi(u)},\mr e^{-t\xi(v)}]\\
\nonumber
&=\sum_{u,v\in\mbb Z^d}\mr e^{-tV(u)-tV(v)}\mr e^{t^2\ga(0)}\left(\mr e^{t^2\ga(u-v)}-1\right)\\
\label{Equation: Heuristic 2}
&\approx t^2\,\sum_{u,v\in\mbb Z^d}\mr e^{-tV(u)-tV(v)}\ga(u-v),
\end{align}
where the last line follows from a Taylor expansion.
A bound of this type can be achieved in the general case thanks to our assumption that
$\xi$'s tails are not worse than exponential. We refer to Proposition \ref{Proposition: Variance Formula}
for the general form of the variance formula. See Lemmas
\ref{Lemma: Variance Upper Bound 1} and \ref{Lemma: Variance Upper Bound 2}
for quantitative bounds on the vanishing of the covariance of the exponential random field $\mr e^{-t\xi}$
as $t\to0$ in terms of the strength of $\xi$'s covariance.

Our third and final step is to identify conditions such that the quantity
\begin{align}
\label{Equation: Heuristic 3}
\sum_{u,v\in\mbb Z^d}\mr e^{-tV(u)-tV(v)}\ga(u-v)
\end{align}
does not blow up at a faster rate than $t^{-2}$ as $t\to0$. As advertised
in our informal statement, this depends on the growth rate of the potential $V$
and the decay rate (if any) of the covariance $\ga$ at infinity. To give an
illustration of how this is carried out in this paper, we consider the two simplest
(and most extreme) cases of covariance structure:
\begin{enumerate}
\item $\big(\xi(v)\big)_{v\in\mbb Z^d}$ are i.i.d., i.e., $\ga(v)=0$ whenever $v\neq 0$; and
\item $\big(\xi(v)\big)_{v\in\mbb Z^d}$ are all equal to each other, i.e., $\ga(v)=\ga(0)$ for all $v\in\ms V$.
\end{enumerate}
The quantity \eqref{Equation: Heuristic 3} then becomes
\[\sum_{u,v\in\mbb Z^d}\mr e^{-tV(u)-tV(v)}\ga(u-v)=\begin{cases}
\displaystyle
\ga(0)\sum_{v\in\mbb Z^d}\mr e^{-2tV(v)}&\text{i.i.d. case,}
\vspace{10pt}\\
\displaystyle
\ga(0)\left(\sum_{v\in\mbb Z^d}\mr e^{-tV(v)}\right)^2&\text{all equal case.}
\end{cases}\]
If we assume that $V(v)\gg\msf d(0,v)^\al$ for some $\al>0$,
then for any $\theta>0$ we have that
\begin{align}
\label{Equation: Heuristic 4}
\sum_{v\in\mbb Z^d}\mr e^{-\theta tV(v)}\ll\sum_{v\in\mbb Z^d}\mr e^{-\theta t\msf d(0,v)^\al}=\sum_{n\in\mbb N\cup\{0\}}\msf c_n(0)\mr e^{-\theta tn^\al},
\end{align}
where we recall that $\msf c_n(0)$ denotes for every $n\in\mbb N$ the number of vertices in $\ms G$ such that $\msf d(0,v)=n$.
For the $d$-dimensional integer lattice $\mbb Z^d$, it is easy to check that there exists a constant $C>0$ such that $\msf c_n(0)\leq Cn^{d-1}$
for every $n\in\mbb N$, whence \eqref{Equation: Heuristic 4} yields
\begin{align}
\label{Equation: Heuristic 5}
\sum_{v\in\mbb Z^d}\mr e^{-\theta tV(v)}\ll\sum_{n\in\mbb N\cup\{0\}}n^{d-1}\mr e^{-\theta tn^\al}\approx\int_0^\infty x^{d-1}\mr e^{-\theta tx^\al}\d x=O(t^{-d/\al}).
\end{align}

Summarizing our argument so far in \eqref{Equation: Heuristic 1}--\eqref{Equation: Heuristic 5},
we are led to the $t\to0$ asymptotic
\[\mbf{Var}\big[\mr{Tr}[\mr e^{-t H}]\big]
\ll
\begin{cases}
t^{2-d/\al}&\text{i.i.d. case,}\\
t^{2-2d/\al}&\text{all equal case.}
\end{cases}\]
Thus, the eigenvalue point process of $H$ is proved to be number rigid if $V(v)\gg\msf d(0,v)^{d/2}$ in the i.i.d case
and $V(v)\gg\msf d(0,v)^d$ in the all equal case. If $\ga$ has a less extreme decay rate
(such as $\ga(v)=O(\msf d(0,v)^{-\be})$ as $\msf d(0,v)\to\infty$ for some $\be>0$), then the eigenvalue point process of $H$ is number rigid if $V(v)\gg\msf d(0,v)^\al$ for some $d/2\leq\al\leq d$, where the exact value
of $\al$ depends on $\ga$'s decay rate. We refer to Theorems \ref{Theorem: Upper} and \ref{Theorem: Rigidity}
for the details.

\section{Main Results}\label{sec:Main Result}

\subsection{Basic Definitions and Notations}

We begin by introducing basic/standard notations
that will be used throughout the paper.

\begin{notation}[Function Spaces]
We use $\ell^p(\ms V)$ to denote the space of real-valued
absolutely $p$-summable (or bounded if $p=\infty$) functions on $\ms V$; we
denote the associated norm by $\|\cdot\|_{p}$.
We use $\langle\cdot,\cdot\rangle$ to denote the inner
product on $\ell^2(\ms V)$.
Given a subset $\ms U\subset\ms V$, we denote
\[\ell^p_\ms U(\ms V):=\{f\in\ell^p(\ms V):f(u)=0\text{ for every }u\in\ms U\}.\]
\end{notation}

\begin{notation}[Operator Theory]\label{def:Multiplicities}
Given a linear operator $T$ on $\ell^2_\ms U(\ms V)$ (or a dense domain $D(T)\subset\ell^2_\ms U(\ms V)$), we use $\si(T)$ to denote its spectrum,
and $\si_p(T)\subset\si(T)$ to denote its point spectrum. If $T$ is bounded, we denote its operator norm by
\[\|T\|_{\mr{op}}:=\sup_{f\in\ell^2_\ms U(\ms V),~\|f\|_{2}=1}\|Tf\|_{2}.\]
We use $\mf R(z,T):=(T-z)^{-1}$
to denote the resolvent of $T$ for all $z\in\mbb C\setminus\si(T)$.
If $\la$ is an isolated eigenvalue of $T$, then we let
\[m_a(\la,T):=\dim\left(\mr{rg}\left(\frac1{2\pi\mr i}\oint_{\Ga_\la}\mf R(z,T)\d z\right)\right)\]
denote the algebraic multiplicity of $\la$, where $\dim$ denotes the dimension of a linear space,
$\mr{rg}$ denotes the range of an operator, and
$\Ga_\la$ denotes a Jordan curve that encloses $\la$ and excludes the remainder
of the spectrum of $T$.
\end{notation}

\begin{definition}[Rigidity]
\label{Definition: Rigidity}
Let $\mc X=\sum_{k\in\mbb N}\de_{\la_k}$ be an infinite point process on $\mbb C$.
We say that $\mc X$ is real-bounded below by a random variable $\om\in\mbb R$
if $\Re(\la_k)\geq\om$ almost surely for every $k\in\mbb N$.
We say that such a point process is number rigid if for every Borel set $B\subset\mbb C$
such that $B\subset(-\infty,\de]+\mr i[-\tilde \de,\tilde \de]$
for some $\de,\tilde\de>0$, the random variable $\mc X(B)$ is measurable with
respect to the completion (under the law of the point process $\mc X$) of the sigma algebra generated by the set
\[\big\{\mc X(A):A\subset\mbb C\text{ is Borel and }B\cap A=\varnothing\big\}.\]
\end{definition}

\begin{remark}
In previous works in the literature, it is most common to define number rigidity
as the requirement that $\mc X(B)$ is measurable with respect to the configuration in $\mbb C\setminus B$
for every bounded Borel set $B$. This is in part due to the fact that most point processes
that have been proved to be number rigid thus far are such that $\mc X(B)=\infty$ almost surely
whenever $B$ is unbounded.

That being said, the fact that we are considering the spectrum of Schr\"odinger
operators whose potentials have a strong growth at infinity means that we are considering
eigenvalue point processes that are real-bounded below, in which case a more general
notion of number rigidity makes sense. We note that a similarly generalized notion of rigidity appeared
in the work of Bufetov on the stochastic Airy operator in \cite[Proposition 3.2]{Buf16}.
\end{remark}

\subsection{Markov Process}

Next, we introduce the Markov processes on the graph
$\ms G$ that generate
our random operators, as well as some of the notions we
need to describe them. We recall that $\ms G=(\ms V,\ms E)$ is a countably infinite
connected graph with uniformly bounded degrees and a root $0\in\ms V$.

\begin{definition}[Markov Process]
Let $\Pi:\ms V\times\ms V\to[0,1]$ be a matrix such that
\begin{enumerate}
\item $\Pi$ is stochastic, that is, for every $u\in \ms V$,
\[\sum_{v\in\ms V}\Pi(u,v)=1;\]
\item $\Pi(v,v)=0$ for all $v\in\ms V$; and
\item If $(u,v)\not\in\ms E$, then $\Pi(u,v)=\Pi(v,u)=0$.
\end{enumerate}
Let $q:\ms V\to(0,\infty)$ be a positive vector and let
$X:[0,\infty)\to\ms V$ denote the continuous-time Markov process on $\ms V$
defined as follows. If $X$ is in
state $u\in\ms V$, it waits for a random time with
an exponential distribution with rate $q(u)$, and then jumps
to another state $v\neq u$ with probability $\Pi(u,v)$, independently of the wait time.
Once at the new state, $X$ repeats this procedure independently of all previous jumps.
\end{definition}

\begin{remark}
We note that condition (3) in the above definition implies that
$X$ is a Markov process on the graph $\ms G$, in the sense that
jumps can only occur between vertices that are connected by edges.
\end{remark}

\begin{notation}
\label{Notation: Markov Probability and Expectation}
For every $v\in\ms V$, we use
$X^v$ to denote the process $X$
conditioned on the starting point $X(0)=v$.
We use $\mbf P^v$ to denote the
law of $X^v$, and $\mbf E^v$ to denote expectation
with respect to $\mbf P^v$.
\end{notation}

We assume throughout that the Markov process $X$
and the graph $\ms G$ satisfy the following.

\begin{assumption}[Graph Geometry and Jump Rates]
\label{Assumption: Graph}
The following two conditions hold:
\begin{enumerate}
\item There exists constants $d\geq1$ and $\mf c>0$ such that
\begin{align}
\label{Equation: Coordination Sequence}
\sup_{v\in\ms V}\msf c_n(v):=\sup_{v\in\ms V}|\{u\in\ms V:\msf d(u,v)=n\}|\leq \mf c\,n^{d-1}\qquad\text{for all }n\in\mbb N\cup\{0\},
\end{align}
recalling that $\msf d$ is the graph distance in $\ms G$, that is, $\msf d(u,v)$
is the length of the shortest path (in terms of number of edges) connecting
$u$ and $v$, and with the convention that $\msf d(v,v)=0$ for all $v\in\ms V$.
\item $X$ has uniformly bounded jump rates, that is,
\[\displaystyle\mf q:=\sup_{v\in\ms V}q(v)<\infty.\]
\end{enumerate}
\end{assumption}

\begin{remark}
We note that the assumption \eqref{Equation: Coordination Sequence}
simultaneously takes care of the requirement that $\ms G$ has uniformly
bounded degrees (since $\msf c_1(v)=\deg(v)$) and of the asymptotic
growth rate \eqref{Equation: Coordination 1} stated in our informal theorem.
\end{remark}

\subsection{Feynman-Kac Kernel}

We are now in a position to introduce the central objects of study of this paper,
namely, the Feynman-Kac semigroups of the Schr\"odinger operators we
are interested in.

\begin{notation}[Local Time]
For every $t\geq0$, we let $L_t:\ms V\to[0,t]$ denote the local time of $X$:
\[L_t(v):=\int_0^t\mbf 1_{\{X(s)=v\}}\d s,\qquad v\in\ms V.\]
\end{notation}

\begin{definition}[Potential and Noise]
Let $V:\ms V\to\mbb R\cup\{\infty\}$ be a deterministic function,
and let $\xi:\ms V\to\mbb R$ be a random function.
We denote the set
\begin{align}
\label{Equation: Dirichlet Set}
\ms Z:=\{v\in\ms V:V(v)=\infty\}.
\end{align}
\end{definition}

Throughout, we make the following assumptions on the noise and potential.

\begin{assumption}[Potential Growth and Noise Tails]
\label{Assumption: Potential and Noise}
There exists $\al>0$ such that
\begin{align}
\label{Equation: Potential Growth}
\liminf_{\msf d(0,v)\to\infty}\frac{V(v)}{\msf d(0,v)^\al}=\infty.
\end{align}
Moreover, $\xi$ satisfies the following conditions:
\begin{enumerate}
\item $\mbf E[\xi(v)]=0$ for every $v\in\ms V$.
\item There exists $\mf m>0$ such that
for every $p\in\mbb N$,
\begin{align}
\label{Equation: Exponential Moments}
\sup_{v\in\ms V}\mbf E[|\xi(v)|^p]\leq p!\mf m^p.
\end{align}
\end{enumerate}
\end{assumption}

In the sequel, it will be useful to characterize noises in terms
of the decay rate of their covariances. For this purpose, we make
the following definition.

\begin{definition}[covariance decay]
\label{Definition: covariance decay}
We say that $\xi$ has covariance decay of order (at least) $\be>0$ if there
exists a constant $\mf C>0$ such that
\begin{align}
\label{Equation: covariance decay of Order Two}
\big|\mbf E[\xi(u)\xi(v)]\big|\leq\mf C\,\big(\msf d(u,v)+1\big)^{-\be}
\end{align}
for every $u,v\in\ms V$, and such that
\begin{align}
\label{Equation: covariance decay of Order Three}
\big|\mbf E[\xi(u)\xi(v)\xi(w)]\big|\leq\mf C\,\min_{a,b\in\{u,v,w\}}\big(\msf d(a,b)+1\big)^{-\be}
\end{align}
for every $u,v,w\in\ms V$.
\end{definition}

\begin{definition}[Feynman-Kac Kernel]
Define the Feynman-Kac kernel
\begin{align}
\label{Equation: Kernel}
K_t(u,v):=
\mbf E^u\left[\mr e^{-\langle L_t,V+\xi\rangle}\mbf 1_{\{X(t)=v\}}\right],\qquad u,v\in\ms V,
\end{align}
where we assume that $X$ is independent of $\xi$,
and that $\mbf E^v$ denotes the expectation with respect to
the Markov process $X^v$, conditional on $\xi$.
We denote the trace of $K_t$ as
\[\mr{Tr}[K_t]:=\sum_{v\in\ms V}K_t(v,v).\]
\end{definition}

\begin{remark}
In the above definition,
we use the convention that $\mr e^{-\infty}:=0$ whenever $V(v)=\infty$,
in particular, $K_t(u,v)=0$ whenever $u\in\ms Z$ or $v\in\ms Z$.
\end{remark}

\subsection{Main Results: Variance Upper Bound and Rigidity}

We now state our main results. First, we have the following sufficient condition
for the vanishing of the variance of the trace of $K_t$ as $t\to0$:

\begin{theorem}
\label{Theorem: Upper}
Suppose that Assumptions \ref{Assumption: Graph} and \ref{Assumption: Potential and Noise} hold.
In order to have
\[\lim_{t\to0}\mbf{Var}\big[\mr{Tr}[K_t]\big]=0,\]
it is sufficient that the constant $\al$ in \eqref{Equation: Potential Growth}
satisfies the following:
\begin{enumerate}
\item if $\xi$ has covariance decay of order $\be>0$, then
\begin{align}
\label{Equation: Polynomial Growth Condition}
\al\begin{cases}
\geq d/2&\text{when }\be>d,\\
>d/2&\text{when }\be=d,\\
\geq d-\be/2&\text{when }\be<d;
\end{cases}
\end{align}
\item otherwise, $\al\geq d$.
\end{enumerate}
\end{theorem}

As a consequence of the above theorem, we have the following result,
which states some properties of the infinitesimal generator of $K_t$, including number rigidity.

\begin{theorem}
\label{Theorem: Rigidity}
Suppose that Assumptions \ref{Assumption: Graph} and \ref{Assumption: Potential and Noise} hold,
and that we take the constant $\al$ in \eqref{Equation: Potential Growth} as in Theorem \ref{Theorem: Upper}.
The following conditions hold almost surely.
\begin{enumerate}
\item For every $t>0$, $K_t$ is a trace class linear operator
on $\ell^2_\ms Z(\ms V)$. There exists a random variable $\om\leq0$
such that $\|K_t\|_{\mr{op}}\leq\mr e^{-\om t}$ for all $t>0$.
\item The family of operators
$(K_t)_{t>0}$ is a strongly continuous semigroup
on $\ell^2_\ms Z(\ms V)$.
\item The infinitesimal generator
\begin{align}
\label{Equation: Infinitesimal Generator}
H:=\lim_{t\to0}\frac{K_0-K_t}{t}
\end{align}
is closed on some dense domain $D(H)\subset\ell^2_\ms Z(\ms V)$,
and its action on functions is given by the following matrix:
\begin{align}
\label{Equation: Schrodinger Generator}
H(u,v):=\begin{cases}
-q(u)\Pi(u,v)&\text{if }u\neq v\text{ and }u,v\not\in\ms Z,\\
q(u)+V(u)+\xi(u)&\text{if }u=v\text{ and }u\not\in\ms Z,\\
0&\text{if }u\in\ms Z\text{ or }v\in\ms Z.
\end{cases}
\end{align}
(So, if $f\in D(H)$, then $f(v)=0$ for every $v\in\ms Z$.)
\end{enumerate}
In particular, almost surely,
$H$ has a pure point spectrum without accumulation point,
and the eigenvalue point process (counting algebraic multiplicities)
\begin{align}
\label{Equation: Eigenvalue Point Process}
\mc X_H:=\sum_{\la\in\si(H)} m_a(\la,H)\,\de_\la
\end{align}
is real-bounded below by $\om$
and number rigid in the sense of Definition \ref{Definition: Rigidity}.
\end{theorem}

\subsection{Questions of Optimality}
\label{Section: Optimality}

In this section, we study the optimality of the growth assumptions
we make on $V$ in Theorem \ref{Theorem: Upper} by considering three
counterexamples.

\begin{theorem}\label{Theorem: Lower}
Suppose that $X$ is the nearest-neighbor symmetric random walk on the integer lattice $\mbb Z^d$,
that $V(v):=\msf d(0,v)^\de$ for some $\de>0$, and that $\xi$ is a centered stationary Gaussian
process whose covariance function $\ga(v):=\mbf E[\xi(v)\xi(0)]$ is nonnegative.
If one of the following conditions hold:
\begin{enumerate}
\item $\de\leq d/2$ and $\ga(v)=\mbf 1_{\{v=0\}}$;
\item $\de\leq d-\be/2$ for some $0<\be<d$, and there exists a constant $\mf L>0$
such that $\ga(v)\geq\mf L\big(\msf d(0,v)+1\big)^{-\be}$ for every $v\in\ms V$; or
\item $\de\leq d$ and $\inf_{v\in\mbb Z^d}\ga(v)>\mf L$ for some constant $\mf L>0$;
\end{enumerate}
then we have the variance lower bound
\[\liminf_{t\to0}\mbf{Var}\big[\mr{Tr}[K_t]\big]>0.\]
\end{theorem}

Thus, given that $\msf c_n(v)\asymp n^{d-1}$ as $n\to\infty$ on $\mbb Z^d$,
if one is interested in providing a general sufficient condition for number rigidity on graphs
using semigroups, then Theorem \ref{Theorem: Upper} is essentially the optimal result
one could hope for.

\begin{remark}
An examination of the proof of Theorem \ref{Theorem: Lower}
reveals that similar lower bounds can be proved for more general
examples with little effort; we restrict our attention to this elementary setting for
simplicity of exposition.
\end{remark}

\section{Proof of Theorem \ref{Theorem: Upper}}
\label{sec: Proof of Upper}

Throughout this section, we suppose that Assumptions
\ref{Assumption: Graph} and \ref{Assumption: Potential and Noise} hold.
This section is organized as follows: In Section \ref{Section: Outline 2}, we outline the
main steps of the proof of Theorem \ref{Theorem: Upper}.
That is, we state a number of technical propositions and lemmas,
which we then use to prove Theorem~\ref{Theorem: Upper}. Then, in Sections
\ref{Section: Proof of Variance Formula}--\ref{sec:PrLem3}, we prove the technical results stated
Section \ref{Section: Outline 2}, thus wrapping-up the proof
of Theorem~\ref{Theorem: Upper}.

\subsection{Proof Outline}
\label{Section: Outline 2}

\subsubsection{Step 1. Variance Formula and First Bound}

We begin with some notations.
\begin{notation}
\label{Notation: Conditional Expectation}
Let us denote by $(\Om_\xi,\mbf P_\xi)$ the probability space on which $\xi$
is defined. Let $Y$ be any random element that is independent of $\xi$,
and let $F$ be any measurable function. We denote the random variable
\[\mbf E_\xi\big[F(\xi,Y)\big]:=\int_{\Om_\xi} F(x,Y)\d\mbf P_\xi(x);\]
that is, $\mbf E_\xi$ is the conditional expectation with respect to
$\xi$, given $Y$. Then, for measurable functions $F$ and $G$,
we denote the random variable
\[\mbf{Cov}_\xi\big[F(\xi,Y),G(\xi,Y)\big]:=\mbf E_\xi\big[F(\xi,Y)G(\xi,Y)\big]-\mbf E_\xi\big[F(\xi,Y)\big]\mbf E_\xi\big[G(\xi,Y)\big].\]
\end{notation}
Our main tool in the proof of Theorem \ref{Theorem: Upper} is the following variance formula:

\begin{proposition}
\label{Proposition: Variance Formula}
For every $u,v\in\ms V$,
we let $X^u$ and $\tilde X^v$ be independent
copies of the Markov process $X$ started from $u$ and $v$ respectively.
We assume that $X^u$ and $\tilde X^v$ are independent of the
noise $\xi$, and
we denote their local times as
\[L^u_t(w):=\int_0^t\mbf 1_{\{X^u(s)=w\}}\d s
\qquad\text{and}\qquad
\tilde L^v_t(w):=\int_0^t\mbf 1_{\{\tilde X^v(s)=w\}}\d s\]
for all $w\in\ms V$. It holds that
\[\mbf{Var}\big[\mr{Tr}[K_t]\big]=\sum_{u,v\in\ms V}\mbf E\left[\mr e^{-\langle L^u_t+\tilde L^v_t,V\rangle}
\mbf{Cov}_\xi\left[\mr e^{-\langle L^u_t,\xi\rangle},\mr e^{-\langle\tilde L^v_t,\xi\rangle}\right]\mbf 1_{\{X^u(t)=u,\tilde X^v(t)=v\}}\right].\]
\end{proposition}

The proof of this proposition, which we provide in Section \ref{Section: Proof of Variance Formula} below,
is essentially a direct consequence of the definition of $K_t$ in \eqref{Equation: Kernel}.
In order to find sufficient conditions for $\mbf{Var}\big[\mr{Tr}[K_t]\big]\to0$ as $t\to0$ using this formula, it is
convenient to control the contributions coming from $V$ and $\xi$ separately. To this end,
we use H\"older's inequality, as well as the elementary fact that $\mbf 1_E\leq 1$ for every event $E$,
which yields
\begin{multline*}
\mbf E\Big[\mr e^{-\langle L^u_t+\tilde L^v_t,V\rangle}
\mbf{Cov}_\xi\big[\mr e^{-\langle L^u_t,\xi\rangle},\mr e^{-\langle\tilde L^v_t,\xi\rangle}\big]\mbf 1_{\{X^u(t)=u,\tilde X^v(t)=v\}}\Big]\\
\leq\mbf E\Big[\mr e^{-2\langle L^u_t+\tilde L^v_t,V\rangle}
\Big]^{1/2}\mbf E\Big[\mbf{Cov}_{\xi}\Big[\mr e^{-\langle L^u_t,\xi\rangle},\mr e^{-\langle\tilde L^v_t,\xi\rangle}\Big]^2\Big]^{1/2}
\end{multline*}
for every fixed $u,v\in\ms V$.
Then, by summing both sides of the above inequality over $u,v \in \ms V$,
we obtain our first upper bound for the variance:
\begin{align}
\label{Equation: Variance Upper Bound Holder}
\mbf{Var}\big[\mr{Tr}[K_t]\big]\leq
\sum_{u,v\in\ms V}\mbf E\Big[\mr e^{-2\langle L^u_t+\tilde L^v_t,V\rangle}
\Big]^{1/2}\mbf E\Big[\mbf{Cov}_{\xi}\Big[\mr e^{-\langle L^u_t,\xi\rangle},\mr e^{-\langle\tilde L^v_t,\xi\rangle}\Big]^2\Big]^{1/2}.
\end{align}

\subsubsection{Step 2. Controlling the Contributions from $\xi$ and $V$}

We now state the technical results that we use to control the right-hand side of \eqref{Equation: Variance Upper Bound Holder}.
Our first such result is as follows:

\begin{lemma}
\label{Lemma: Variance Upper Bound 1}
Recall the definition of the constant $\mf m>0$ in \eqref{Equation: Exponential Moments}.
There exists a constant $C_1>0$ (which only depends on $\mf m$) such that
for every $t<1/C_1$, one has
\begin{align*}
\sup_{u,v\in\ms V}\mbf E\Big[\mbf{Cov}_{\xi}\big[\mr e^{-\langle L^u_t,\xi\rangle},\mr e^{-\langle\tilde L^v_t,\xi\rangle}\big]^2\Big]^{1/2}\leq C_1t^2.
\end{align*}
\end{lemma}

The proof of Lemma \ref{Lemma: Variance Upper Bound 1}, which we provide in Section \ref{sec:PrLem1},
follows from estimating expectations of the form $\mbf E^v\big[\mr e^{-\theta\langle L_t,\xi\rangle}\big]$
using our assumption that $\xi$'s tails are not worse than exponential (i.e., \eqref{Equation: Exponential Moments}).
Next, we have the following result, which provides a tighter decay rate
in the case where $\xi$ has covariance decay:

\begin{lemma}
\label{Lemma: Variance Upper Bound 2}
Suppose that $\xi$ has covariance decay of order $\beta$,
as per Definition \ref{Definition: covariance decay}.
Recall the definitions of the constants $\mf q$, $\mf m$, and $\mf C$
in Assumption \ref{Assumption: Graph} (3), \eqref{Equation: Exponential Moments},
\eqref{Equation: covariance decay of Order Two}, and \eqref{Equation: covariance decay of Order Three}.
There exists a constant $C_2>0$
(which only depends on $\mf q$, $\mf m$, $\mf C$, and $\be$)
such that for every $t<1/C_2$ and $u,v\in\ms V$, one has
\[\mbf E\left[\mbf{Cov}_\xi\left[\mr e^{-\langle L^u_t,\xi\rangle},\mr e^{-\langle\tilde L^v_t,\xi\rangle}\right]^2\right]^{1/2}\leq C_2\left(t^2\big(\msf d(u,v)+1\big)^{-\be}+t^4\right).\]
\end{lemma}

Lemma \ref{Lemma: Variance Upper Bound 2} is proved in Section \ref{sec:PrLem2}.
The proof of this lemma is rather more subtle than that of Lemma \ref{Lemma: Variance Upper Bound 1},
and depends on a careful control of how much $X^u$ and $\tilde X^v$ deviate from their
respective starting points $u$ and $v$. We note that the uniform upper bound on
the jump rates of $X$ in Assumption \ref{Assumption: Graph} (3) is crucial for this lemma.

\begin{remark}
The proofs of Lemmas \ref{Lemma: Variance Upper Bound 1} and \ref{Lemma: Variance Upper Bound 2} both rely on some elementary formulas and estimates of the moment generating functions of the noises and their covariances, which will be stated and proved in Section~\ref{Section: Moment Generating Functions}.
\end{remark}

With Lemmas \ref{Lemma: Variance Upper Bound 1} and \ref{Lemma: Variance Upper Bound 2}
in hand, it now only remains to control the contribution of the potential $V$ in \eqref{Equation: Variance Upper Bound Holder}.
For this, we have the following result:

\begin{lemma}
\label{Lemma: Variance Upper Bound 3}
Recall the definition of $d\geq1$ and $\mf c>0$ in \eqref{Equation: Coordination Sequence}.
Suppose that we can find some constants $\ka,\mu>0$ such that
\begin{align}
\label{Equation: V Pointwise Lower Bound}
V(v)\geq\big(\ka\,\msf d(0,v)\big)^\al-\mu,\qquad v\in\ms V.
\end{align}
Then, there exists a constant $C_3>0$
(which only depends on $\al$, $\be$, $d$, and $\mf c$)
such that
\begin{align}
\label{Equation: Variance Upper Bound 3.1}
&\limsup_{t\to0}t^{2d/\al}\sum_{u,v\in\ms V}
\mbf E\Big[\mr e^{-2\langle L^u_t+\tilde L^v_t,V\rangle}\Big]^{1/2}
\leq C_3\ka^{-2d};\\
\label{Equation: Variance Upper Bound 3.2}
&\limsup_{t\to0}t^{(2d-\be)/\al}\sum_{u,v\in\ms V}
\mbf E\Big[\mr e^{-2\langle L^u_t+\tilde L^v_t,V\rangle}\Big]^{1/2}\big(\msf d(u,v)+1\big)^{-\be}
\leq C_3\ka^{-2d+\be}
\end{align}
for every $0<\be<d$; and
\begin{align}
\label{Equation: Variance Upper Bound 3.3}
\limsup_{t\to0}t^{d/\al}\sum_{u,v\in\ms V}
\mbf E\Big[\mr e^{-2\langle L^u_t+\tilde L^v_t,V\rangle}\Big]^{1/2}\big(\msf d(u,v)+1\big)^{-\be}
\leq C_3\ka^{-d}
\end{align}
for every $\be>d$.
\end{lemma}

Lemma \ref{Lemma: Variance Upper Bound 3}, which is proved in Section \ref{sec:PrLem3},
follows the strategy outlined in \eqref{Equation: Heuristic 4} and \eqref{Equation: Heuristic 5}:
The first step of the proof of Lemma \ref{Lemma: Variance Upper Bound 3} relies on a rigorous
implementation of the intuition that, for very small $t>0$, one expects that
\begin{align}
\label{Equation: Lemma 3 Intuition}
\mbf E\big[\mr e^{-2\langle L^u_t+\tilde L^v_t,V\rangle}\big]^{1/2}\approx\mr e^{-tV(u)-tV(v)}.
\end{align}
This once again relies on controlling how much  $X^u$ and $\tilde X^v$ deviate from their starting points.
Once a quantitative version of \eqref{Equation: Lemma 3 Intuition} is established, 
we can then use \eqref{Equation: V Pointwise Lower Bound}, which allows to control
$\mbf E\big[\mr e^{-2\langle L^u_t+\tilde L^v_t,V\rangle}\big]^{1/2}$
in terms of quantities that only depend on the geometry of $\ms G$ (more precisely, the graph distance).
We then wrap up the proof of the lemma by using the upper bound on the coordination sequences in
\eqref{Equation: Coordination Sequence}, in similar fashion to \eqref{Equation: Heuristic 5}.

\subsubsection{Step 3. Conclusion of Proof}

We now combine the technical results stated above to conclude the proof of Theorem
\ref{Theorem: Upper}.
By applying Lemmas \ref{Lemma: Variance Upper Bound 1} and \ref{Lemma: Variance Upper Bound 2}
to our upper bound \eqref{Equation: Variance Upper Bound Holder}, we get that for every $t<1/C_1$, one has
\begin{align}
\label{Equation: Main Proof General}
\mbf{Var}\big[\mr{Tr}[K_t]\big]\leq C_1 t^2\sum_{u,v\in\ms V}\mbf E\Big[\mr e^{-2\langle L^u_t+\tilde L^v_t,V\rangle}\Big]^{1/2},
\end{align}
and if $\xi$ has covariance decay of order $\be>0$, then for every $t<1/C_2$, one has
\begin{multline}
\label{Equation: Main Proof Decay}
\mbf{Var}\big[\mr{Tr}[K_t]\big]\leq C_2 t^2\sum_{u,v\in\ms V}\mbf E\Big[\mr e^{-2\langle L^u_t+\tilde L^v_t,V\rangle}\Big]^{1/2}\big(\msf d(u,v)+1\big)^{-\be}\\
+C_2t^4\sum_{u,v\in\ms V}\mbf E\Big[\mr e^{-2\langle L^u_t+\tilde L^v_t,V\rangle}\Big]^{1/2}.
\end{multline}
Thanks to our growth assumption in \eqref{Equation: Potential Growth}, for any choice of $\ka>0$, we know that there exists
a large enough $\mu>0$ so that \eqref{Equation: V Pointwise Lower Bound} holds. We may then complete the proof
of Theorem \ref{Theorem: Upper} by an application of Lemma \ref{Lemma: Variance Upper Bound 3}. We do this on a case-by-case
basis:

Suppose first that $\xi$ has covariance decay of order $0<\be<d$ and that $\al\geq d-\be/2>d/2$. Then,
the fact that $2-(2d-\be)/\al\geq0$ implies by \eqref{Equation: Variance Upper Bound 3.2} that
\begin{multline*}
\limsup_{t\to0}t^2\sum_{u,v\in\ms V}\mbf E\Big[\mr e^{-2\langle L^u_t+\tilde L^v_t,V\rangle}\Big]^{1/2}\big(\msf d(u,v)+1\big)^{-\be}\\
=\limsup_{t\to0}t^{2-(2d-\be)/\al}t^{(2d-\be)/\al}\sum_{u,v\in\ms V}\mbf E\Big[\mr e^{-2\langle L^u_t+\tilde L^v_t,V\rangle}\Big]^{1/2}\big(\msf d(u,v)+1\big)^{-\be}
\leq C_3\ka^{-2d+\be};
\end{multline*}
and the fact that $4-2d/\al>0$ implies by \eqref{Equation: Variance Upper Bound 3.1} that
\begin{multline}
\label{Equation: Main Proof Decay 2}
\limsup_{t\to0}t^4\sum_{u,v\in\ms V}\mbf E\Big[\mr e^{-2\langle L^u_t+\tilde L^v_t,V\rangle}\Big]^{1/2}\\
=\limsup_{t\to0}t^{4-2d/\al}t^{2d/\al}\sum_{u,v\in\ms V}\mbf E\Big[\mr e^{-2\langle L^u_t+\tilde L^v_t,V\rangle}\Big]^{1/2}
=0.
\end{multline}
Combining this with \eqref{Equation: Main Proof Decay} implies that
\[\limsup_{t\to0}\mbf{Var}\big[\mr{Tr}[K_t]\big]\leq C_2C_3\ka^{-2d+\be},\]
where we recall that $C_2,C_3>0$ do not depend on $\ka$ or $\mu$. Since
\eqref{Equation: V Pointwise Lower Bound} holds for any choice of $\ka>0$, we can take
$\ka\to\infty$, which then yields $\mbf{Var}\big[\mr{Tr}[K_t]\big]\to0$ as $t\to0$.

Next, suppose that $\xi$ has covariance decay of order $\be=d$ and that $\al>d/2$.
We note that this implies that $\xi$ also has correlation
decay of order $\tilde\be$ for any choice of $0<\tilde\be<d$.
Since $\al>d/2$ implies that $2d-2\al<d$, we can choose $\tilde\be$ close enough to
$d$ so that
$2d-2\al<\tilde\be$, which we can rearrange into $2>(2d-\tilde\be)/\al$.
Thus, \eqref{Equation: Variance Upper Bound 3.2} implies that
\begin{multline*}
\limsup_{t\to0}t^2\sum_{u,v\in\ms V}\mbf E\Big[\mr e^{-2\langle L^u_t+\tilde L^v_t,V\rangle}\Big]^{1/2}\big(\msf d(u,v)+1\big)^{-\be}\\
=\limsup_{t\to0}t^{2-(2d-\tilde\be)/\al}t^{(2d-\tilde\be)/\al}\sum_{u,v\in\ms V}\mbf E\Big[\mr e^{-2\langle L^u_t+\tilde L^v_t,V\rangle}\Big]^{1/2}\big(\msf d(u,v)+1\big)^{-\tilde\be}
=0.
\end{multline*}
Combining this with \eqref{Equation: Main Proof Decay 2}, we directly prove
that $\mbf{Var}\big[\mr{Tr}[K_t]\big]\to0$ as $t\to0$ in this case.

Suppose now that $\xi$ has covariance decay of order $\be>d$ and that $\al\geq d/2$.
Then, the fact that $2-d/\al\geq0$ implies by \eqref{Equation: Variance Upper Bound 3.3} that
\begin{multline*}
\limsup_{t\to0}t^2\sum_{u,v\in\ms V}\mbf E\Big[\mr e^{-2\langle L^u_t+\tilde L^v_t,V\rangle}\Big]^{1/2}\big(\msf d(u,v)+1\big)^{-\be}\\
=\limsup_{t\to0}t^{2-d/\al}t^{d/\al}\sum_{u,v\in\ms V}\mbf E\Big[\mr e^{-2\langle L^u_t+\tilde L^v_t,V\rangle}\Big]^{1/2}\big(\msf d(u,v)+1\big)^{-\be}
\leq C_3\ka^{-d};
\end{multline*}
and the fact that $4-2d/\al\geq0$ implies by \eqref{Equation: Variance Upper Bound 3.1} that
\begin{multline*}
\limsup_{t\to0}t^4\sum_{u,v\in\ms V}\mbf E\Big[\mr e^{-2\langle L^u_t+\tilde L^v_t,V\rangle}\Big]^{1/2}\\
=\limsup_{t\to0}t^{4-2d/\al}t^{2d/\al}\sum_{u,v\in\ms V}\mbf E\Big[\mr e^{-2\langle L^u_t+\tilde L^v_t,V\rangle}\Big]^{1/2}
\leq C_3\ka^{-2d}.
\end{multline*}
Combining this with \eqref{Equation: Main Proof Decay} and taking $\ka\to\infty$ then implies that
$\mbf{Var}\big[\mr{Tr}[K_t]\big]\to0$ as $t\to0$.

Finally, consider the general case where we simply assume that $\al\geq d$.
Then, $2-2d/\al\geq0$, and thus \eqref{Equation: Variance Upper Bound 3.1} implies that
\begin{multline*}
\limsup_{t\to0} t^2\sum_{u,v\in\ms V}\mbf E\Big[\mr e^{-2\langle L^u_t+\tilde L^v_t,V\rangle}\Big]^{1/2}\\
=\limsup_{t\to0} t^{2-2d/\al}t^{2d/\al}\sum_{u,v\in\ms V}\mbf E\Big[\mr e^{-2\langle L^u_t+\tilde L^v_t,V\rangle}\Big]^{1/2}\leq C_3\ka^{-2d}.
\end{multline*}
Since the constants $C_1,C_3>0$ are independent of $\ka$ and $\mu$, combining this with \eqref{Equation: Main Proof General}
and taking $\ka\to\infty$ then implies that $\mbf{Var}\big[\mr{Tr}[K_t]\big]\to0$ as $t\to0$ in this case.
This then completes the proof of Theorem \ref{Theorem: Upper}.

\subsection{Proof of Proposition \ref{Proposition: Variance Formula}}
\label{Section: Proof of Variance Formula}

Since the random walk $X$ is assumed independent of $\xi$, by applying Fubini's theorem
to the definition of $K_t$ in \eqref{Equation: Kernel},
we have that
\[\mbf E\big[\mr{Tr}[K_t]\big]=\sum_{v\in\ms V}\mbf E^v\left[\mr e^{-\langle L_t,V\rangle}\mbf E_\xi\left[\mr e^{-\langle L_t,\xi\rangle}\right]\mbf 1_{\{X(t)=v\}}\right],\]
where we recall the definition of $\mbf E_\xi$ in Notation \ref{Notation: Conditional Expectation}.
Taking the square of this expression, we then get once again by Fubini's theorem that
\[\mbf E\big[\mr{Tr}[K_t]\big]^2=\sum_{u,v\in\ms V}\mbf E\left[\mr e^{-\langle L^u_t+\tilde L^v_t,V\rangle}
\mbf E_\xi\left[\mr e^{-\langle L^u_t,\xi\rangle}\right]\mbf E_\xi\left[\mr e^{-\langle\tilde L^v_t,\xi\rangle}\right]\mbf 1_{\{X^u(t)=u,\tilde X^v(t)=v\}}\right].\]
Thanks to \eqref{Equation: Kernel}, it is easy to check that
\[\mr{Tr}[K_t]^2=\sum_{u,v\in\ms V}\mbf E_{\xi}\left[\mr e^{-\langle L^u_t+\tilde L^v_t,V+\xi\rangle}\mbf 1_{\{X^u(t)=u,\tilde X^v(t)=v\}}\right].\]
Taking the expectation of this expression using Fubini's theorem then leads to
\[\mbf E\big[\mr{Tr}[K_t]^2\big]=\sum_{u,v\in\ms V}\mbf E\left[\mr e^{-\langle L^u_t+\tilde L^v_t,V\rangle}
\mbf E_\xi\left[\mr e^{-\langle L^u_t+\tilde L^v_t,\xi\rangle}\right]\mbf 1_{\{X^u(t)=u,\tilde X^v(t)=v\}}\right].\]
The proof of Proposition \ref{Proposition: Variance Formula} is then simply a matter of subtracting
$\mbf E\big[\mr{Tr}[K_t]\big]^2$ from the above expression for $\mbf E\big[\mr{Tr}[K_t]^2\big]$,
and using the definition of $\mbf{Cov}_\xi$ in Notation \ref{Notation: Conditional Expectation}.

\subsection{Auxiliary results on estimates of moment generating functions} 
\label{Section: Moment Generating Functions}
Before discussing the proofs of Lemma~\ref{Lemma: Variance Upper Bound 1} and Lemma~\ref{Lemma: Variance Upper Bound 2} in the next two subsections, we list here two simple propositions concerning the tail behaviors of the moment generating functions of the noises and their covariances. The first result is a  straightforward consequence of Taylor expansions and Assumption~\ref{Assumption: Potential and Noise} on the tails of the noises. 
\begin{proposition}
	\label{Proposition: Noise Decay 1}
	Under Assumption \ref{Assumption: Potential and Noise},
	for every finitely-supported deterministic functions $f,g:\ms V\to\mbb R$ such that
	$\|f+g\|_{1},\|f\|_{1},\|g\|_{1}\leq1/2\mf m$,
	it holds that
	\begin{align}\label{eq: NoiseDecay1}
		\Big|\mbf E\big[\mr e^{\langle f,\xi\rangle}\big]-1\Big|\leq 2\mf m^2\|f\|_{1}^2
	\end{align}
	and
	\begin{align}\label{eq: NoiseDecay2}
		\big|\mbf{Cov}\big[\mr e^{\langle f,\xi\rangle},\mr e^{\langle g,\xi\rangle}\big]\Big|\leq2\mf m^2\big(\|f+g\|_{1}^2+\|f\|_{1}^2+\|g\|^2_{1}\big)+4\mf m^4\|f\|_{1}^2\|g\|_{1}^2.
	\end{align}
\end{proposition}
\begin{proof}
	For every deterministic function $f:\ms V\to\mbb R$, it follows from
	a straightforward Taylor expansion of the exponential that
	\begin{align}
		\label{Equation: Exponential Moment Expansion}
		\mbf E\left[\mr e^{\langle f,\xi\rangle}\right]=\sum_{p=0}^\infty\frac{1}{p!}\sum_{v_1,\ldots,v_p\in\ms V}\mbf E[\xi(v_1)\cdots\xi(v_p)]f(v_1)\cdots f(v_p),
	\end{align}
	with the convention that the term with $p=0$ above is equal to one.
	Firstly, since $\mbf E[\xi(v)]=0$ for all $v$,
	the term corresponding to $p=1$ in \eqref{Equation: Exponential Moment Expansion} is zero.
	Secondly, thanks to our moment growth assumption $\mbf E[|\xi(v)|^p]\leq p!\mf m^p$, for every $p\geq2$ we have that
	\begin{multline*}
		\left|\sum_{v_1,\ldots,v_p\in\ms V}\mbf E[\xi(v_1)\cdots\xi(v_p)]f(v_1)\cdots f(v_p)\right|\\
		\leq\sum_{v_1,\ldots,v_p\in\ms V}\mbf E[|\xi(v_1)|^p]^{1/p}\cdots\mbf E[|\xi(v_p)|^p]^{1/p}|f(v_1)|\cdots |f(v_p)|
		\leq p!\big(\mf m\|f\|_{1}\big)^p.
	\end{multline*}
	Thus, if $\|f\|_{1}\leq1/2\mf m$, then we have that
	\[\Big|\mbf E\left[\mr e^{\langle f,\xi\rangle}\right]-1\Big|\leq\sum_{p=2}^\infty(\mf m\|f\|_{1})^p=\frac{(\mf m\|f\|_{1})^2}{1-\mf m\|f\|_{1}}\leq2(\mf m\|f\|_{1})^2.\]
	As for the claim regarding the covariance, for any two random variables $Y$ and $Z$,
	we have by the triangle inequality that
	\begin{multline*}
		|\mbf{Cov}[Y,Z]|
		=|\mbf E[YZ]-\mbf E[Y]\mbf E[Z]|\\
		\leq|\mbf E[YZ]-1|-|\mbf E[Y]-1||\mbf E[Z]-1|+|1-\mbf E[Y]|+|1-\mbf E[Z]|
	\end{multline*}
	Thus, whenever $\|f+g\|_{1},\|f\|_{1},\|g\|_{1}\leq1/2\mf m$,
	it follows from \eqref{eq: NoiseDecay1} that
	\[\Big|\mbf{Cov}\big[\mr e^{\langle f,\xi\rangle},\mr e^{\langle g,\xi\rangle}\big]\Big|
	\leq2\mf m^2\big(\|f+g\|_{1}^2+\|f\|_{1}^2+\|g\|^2_{1}\big)+4\mf m^4\|f\|_{1}^2\|g\|_{1}^2,\]
	as desired.
\end{proof}

In cases where we need a more precise control on the covariance, we have the following power series expansion:

\begin{proposition}
	\label{Proposition: General Covariance Formula}
	Suppose that Assumption \ref{Assumption: Potential and Noise} holds.
	For any two finitely supported deterministic functions $f,g:\ms V\to\mbb R$,
	one has
	\[\mbf{Cov}\left[\mr e^{\langle f,\xi\rangle},\mr e^{\langle g,\xi\rangle}\right]=\sum_{p=2}^\infty\frac{\mc A_p(f,g)}{p!},\]
	where, for every $p\geq2$, we denote
	\begin{multline}
		\label{Equation: General Covariance Formula}
		\mc A_p(f,g):=
		\sum_{v_1,\ldots,v_p\in\ms V}\Bigg(\sum_{m=1}^{p-1}{p\choose m}
		\mbf{Cov}[\xi(v_1)\cdots\xi(v_m),\xi(v_{m+1})\cdots\xi(v_p)]\\
		\cdot f(v_1)\cdots f(v_m)g(v_{m+1})\cdots g(v_p)\Bigg).
	\end{multline}
\end{proposition}
\begin{proof}
	Using the same Taylor expansion as in \eqref{Equation: Exponential Moment Expansion},
	we get, on the one hand,
	\begin{align*}
		&\mbf E\left[\mr e^{\langle f+g,\xi\rangle}\right]\\
		&=\sum_{p=0}^\infty\frac{1}{p!}\sum_{v_1,\ldots,v_p\in\ms V}\mbf E[\xi(v_1)\cdots\xi(v_p)]\big(f(v_1)+g(v_1)\big)\cdots\big(f(v_p)+g(v_p)\big)\\
		&=\sum_{p=0}^\infty\frac{1}{p!}\sum_{v_1,\ldots,v_p\in\ms V}\sum_{m=0}^p{p\choose m}\mbf E[\xi(v_1)\cdots\xi(v_p)]f(v_1)\cdots f(v_m)g(v_{m+1})\cdots g(v_p),
	\end{align*}
	and on the other hand
	\begin{align*}
		&\mbf E\left[\mr e^{\langle f,\xi\rangle}\right]\mbf E\left[\mr e^{\langle g,\xi\rangle}\right]\\
		&=\sum_{m_1,m_2=0}^\infty\frac{1}{m_1!m_2!}\Bigg(\sum_{v_1,\ldots,v_{m_1+m_2}\in\ms V}
		\mbf E[\xi(v_1)\cdots\xi(v_{m_1})]\mbf E[\xi(v_{m_1+1})\cdots\xi(v_{m_1+m_2})]\\
		&\hspace{2.5in}\cdot f(v_1)\cdots f(v_{m_1})g(v_{m_1+1})\cdots g(v_{m_1+m_2})\Bigg)\\
		&=\sum_{p=0}^\infty\sum_{m=0}^p\frac{1}{m!(p-m)!}\Bigg(\sum_{v_1,\ldots,v_{p}\in\ms V}
		\mbf E[\xi(v_1)\cdots\xi(v_{m})]\mbf E[\xi(v_{m+1})\cdots\xi(v_{p})]\\
		&\hspace{2.5in}\cdot f(v_1)\cdots f(v_{m})g(v_{m+1})\cdots g(v_{p})\Bigg)\\
		&=\sum_{p=0}^\infty\frac1{p!}\sum_{v_1,\ldots,v_{p}\in\ms V}\Bigg(\sum_{m=0}^p{p\choose m}
		\mbf E[\xi(v_1)\cdots\xi(v_{m})]\mbf E[\xi(v_{m+1})\cdots\xi(v_{p})]\\
		&\hspace{2.5in}\cdot f(v_1)\cdots f(v_{m})g(v_{m+1})\cdots g(v_{p})\Bigg).
	\end{align*}
	We then get the result by subtracting these two expressions.
\end{proof}

\subsection{Proof of Lemma \ref{Lemma: Variance Upper Bound 1}}
\label{sec:PrLem1}

By definition of local time, $\|L_t^u\|_{1}=\|\tilde{L}_t^v\|_{1}= t$,
as well as $\|L_t^u+\tilde{L}_t^v\|_{1}=2t$.
Thus, by \eqref{eq: NoiseDecay2} in Proposition \ref{Proposition: Noise Decay 1}, if $t<1/4\mf m$, then
we have for any $u,v\in \ms V$ that
\[
		\left|\mbf{Cov}_{\xi}\big[\mr e^{-\langle L^u_t,\xi\rangle},\mr e^{-\langle\tilde L^v_t,\xi\rangle}\big]\right|\leq2\mf m^2\big(4t^2+t^2+t^2\big)+4\mf m^4t^4=12\mf m^2t^2+4\mf m^4t^4.
\] 
Since the right-hand side of this inequality is not random, the result then follows
by noting that $t^4\leq t^2$ when $t\leq1$ and taking $C_1:=\max\{1,4\mf m,12\mf m^2,4\mf m^4\}$.

\subsection{Proof of Lemma \ref{Lemma: Variance Upper Bound 2}}
\label{sec:PrLem2}

For every $u,v\in\ms V$ and $t>0$, let us denote by
\[\mf D^{u,v}_t:=\min_{\substack{a,b\in\ms V\\L_t^u(a),\tilde L_t^v(b)\neq0}}\msf d(a,b)\]
the distance between the ranges of $X^u$ and $\tilde X^v$ up to time $t$.
In Section \ref{Section: Covariance Decay Step 1} below we prove the following crude version of Lemma \ref{Lemma: Variance Upper Bound 2}:
For every $t<\min\{1,1/4\mf m\}$ and $u,v\in\ms V$,
\begin{align}
\label{Equation: covariance decay}
\left|\mbf{Cov}_\xi\left[\mr e^{-\langle L_t^u,\xi\rangle},\mr e^{-\langle\tilde L_t^v,\xi\rangle}\right]\right|\leq 2\mf Ct^2(\mf D_t^{u,v}+1)^{-\be}+64\mf m^4t^4.
\end{align}
With this in hand, by Minkowski's inequality, we have that
\begin{align}
\label{Equation: PrLem2 - 1}
\mbf E\left[\mbf{Cov}_\xi\left[\mr e^{-\langle L_t^u,\xi\rangle},\mr e^{-\langle\tilde L_t^v,\xi\rangle}\right]^2\right]^{1/2}\leq
2\mf Ct^2\mbf E\big[(\mf D_t^{u,v}+1)^{-2\be}\big]^{1/2}+64\mf m^4t^4
\end{align}
for every $t<\min\{1,1/4\mf m\}$ and $u,v\in\ms V$.

Next, we control $\mf D^{u,v}_t$ in terms of 
$\msf d(u,v)$. We do this in two cases. Suppose first that $\msf d(u,v)<16$.
In this case, we have the trivial bound
\[\mbf E\big[(\mf D_t^{u,v}+1)^{-2\be}\big]^{1/2}\leq1\leq17^\be\big(\msf d(u,v)+1\big)^{-\be},\]
which, when combined with \eqref{Equation: PrLem2 - 1}, yields
\begin{align}
\label{Equation: PrLem2 - 2}
\mbf E\left[\mbf{Cov}_\xi\left[\mr e^{-\langle L_t^u,\xi\rangle},\mr e^{-\langle\tilde L_t^v,\xi\rangle}\right]^2\right]^{1/2}\leq
2\cdot17^\be\mf Ct^2\big(\msf d(u,v)+1\big)^{-\be}+64\mf m^4t^4
\end{align}
for every $t<\min\{1,1/4\mf m\}$ and $u,v\in\ms V$ such that $\msf d(u,v)<16$.

Suppose then that $\msf d(u,v)\geq16$.
For any $u,v\in \ms V$ and $t>0$, we introduce the event
	\[
		E^{u,v}_t := \left\{\sup_{0\leq s\leq t}\msf d\big(X^u(s),u\big)\leq \frac{\msf d(u,v)}{4}\quad\text{and}\quad\sup_{0\leq s\leq t}\msf d\big(\tilde{X}^v(s),v\big)\leq \frac{\msf d(u,v)}{4}\right\}.
	\]
With this in hand, given that $(\mf D_t^{u,v}+1)^{-\be}\leq1$ and $\sqrt{x+y}\leq\sqrt x+\sqrt y$ for all $x,y\geq0$,
\[\mbf E\big[(\mf D_t^{u,v}+1)^{-2\be}\big]^{1/2}
\leq\mbf E\big[(\mf D_t^{u,v}+1)^{-2\be}\mbf 1_{E^{u,v}_t}\big]^{1/2}+\mbf P\big[(E^{u,v}_t)^c\big]^{1/2}.\]
For any outcome in the event $E^{u,v}_t$, we have by the triangle inequality that
	\[
		\msf d\big(X^u(s),\tilde{X}^v(\tilde s)\big)\geq \msf d(u,v)-\msf d\big(X^u(s),u\big)-\msf d\big(\tilde{X}^v(\tilde s),v\big)\geq  \frac{\msf d(u,v)}{4}
	\]
for every $0\leq s,\tilde s\leq t$. In particular, this means that
$\mf D^{u,v}_t\mbf 1_{E^{u,v}_t}\geq\msf d(u,v)/4$.
In Section \ref{Section: Covariance Decay Step 2} below, we prove that
if $t<\min\{4/\mf q,1/4\mf q\mr e\}$ and $\msf d(u,v)\geq16$,
then
\begin{align}
\label{Equation: Ranges Tail Bound}
\mbf P\big[(E^{u,v}_t)^c\big]^{1/2}\leq \frac{\sqrt 2\,\mf q^2\mr e^2 t^2}{16}.
\end{align}
Combining these bounds with \eqref{Equation: PrLem2 - 1}, we are led to
\begin{multline}
\label{Equation: PrLem2 - 3}
\mbf E\left[\mbf{Cov}_\xi\left[\mr e^{-\langle L_t^u,\xi\rangle},\mr e^{-\langle\tilde L_t^v,\xi\rangle}\right]^2\right]^{1/2}\\
\leq2\cdot 4^\be\mf Ct^2\big(\msf d(u,v)+1\big)^{-\be}+\left(\frac{\sqrt 2\,\mf q^2\mr e^2\mf C}{8}+64\mf m^4\right)t^4
\end{multline}
for all $t<\min\{1,1/4\mf m,4/\mf q,1/4\mf q\mr e\}$ and $u,v\in\ms V$ such that $\msf d(u,v)\geq16$.

With \eqref{Equation: PrLem2 - 2} and \eqref{Equation: PrLem2 - 3} in hand,
in order to prove Lemma \ref{Lemma: Variance Upper Bound 2}, it only remains
to establish \eqref{Equation: covariance decay} and \eqref{Equation: Ranges Tail Bound}.
We do this in the next two subsections.

\subsubsection{Proof of \eqref{Equation: covariance decay}}
\label{Section: Covariance Decay Step 1}

Our main tool to prove \eqref{Equation: covariance decay} consists of the power series expansion
proved in Proposition \ref{Equation: General Covariance Formula}:
\begin{align}
\label{Equation: Covariance Expansion}
\mbf{Cov}_\xi\left[\mr e^{-\langle L_t^u,\xi\rangle},\mr e^{-\langle\tilde L_t^v,\xi\rangle}\right]=\sum_{p=2}^\infty\frac{\mc A_p(-L^u_t,-\tilde L^v_t)}{p!},
\end{align}
where the terms $\mc A_p$ are defined in \eqref{Equation: General Covariance Formula}.
Thanks to our moment growth assumptions in \eqref{Equation: Exponential Moments}, for every $p\geq4$ and $1\leq m\leq p-1$, we have that
\begin{align*}
&\big|\mbf{Cov}[\xi(v_1)\cdots\xi(v_m),\xi(v_{m+1})\cdots\xi(v_p)]\big|\\
&\leq\big|\mbf E[\xi(v_1)\cdots\xi(v_p)]\big|+\big|\mbf E[\xi(v_1)\cdots\xi(v_m)]\mbf E[\xi(v_{m+1})\cdots\xi(v_p)]\big|\\
&\leq\mbf E[|\xi(v_1)|^p]^{1/p}\cdots\mbf E[|\xi(v_p)|^p]^{1/p}\\
&\qquad+\mbf E[|\xi(v_1)|^m]^{1/m}\cdots\mbf E[|\xi(v_m)|^m]^{1/m}\mbf E[|\xi(v_{m+1})|^{p-m}]^{1/(p-m)}\cdots\mbf E[|\xi(v_p)|^{p-m}]^{1/(p-m)}\\
&\leq p!\mf m^p+m!(p-m)!\mf m^p\\
&\leq2p!\mf m^p.
\end{align*}
Therefore, by combining \eqref{Equation: General Covariance Formula}
with the fact that $\sum_{m=0}^p{p\choose m}=2^p$, one has
\[\frac{|\mc A_p(-L_t^u,-\tilde L_t^v)|}{p!}\leq2\mf m^p\sum_{m=1}^{p-1}{p\choose m}\|L_t^u\|_{1}^m\|\tilde L_t^v\|_{1}^{p-m}
\leq 2(2\mf m t)^p.\]
Next, if $\xi$ has covariance decay of order $\be$, then
\eqref{Equation: covariance decay of Order Two} implies that
\begin{multline*}
|\mc A_2(-L_t^u,-\tilde L_t^v)|\leq\sum_{w_1,w_2\in\ms V}\big|\mbf{Cov}[\xi(w_1),\xi(w_2)]\big|L_t^u(w_1)\tilde L_t^v(w_2)\\
\leq\mf C(\mf D_t^{u,v}+1)^{-\be}\|L_t^u\|_{1}\|\tilde L_t^v\|_{1}\leq\mf C t^2(\mf D_t^{u,v}+1)^{-\be}.
\end{multline*}
and similarly \eqref{Equation: covariance decay of Order Three} implies that
\[|\mc A_3(-L_t^u,-\tilde L_t^v)|\leq\mf C t^3(\mf D_t^{u,v}+1)^{-\be}.\]
At this point if we take $t<\min\{1,1/4\mf m\}$, then $t^3\leq t^2$,
and thus it follows from the expansion \eqref{Equation: Covariance Expansion}
and the estimates above that
\begin{multline*}
\left|\mbf{Cov}_\xi\left[\mr e^{-\langle L_t^u,\xi\rangle},\mr e^{-\langle\tilde L_t^v,\xi\rangle}\right]\right|
\leq2\mf Ct^2(\mf D_t^{u,v}+1)^{-\be}+2\sum_{p=4}^\infty(2\mf m t)^p\\
=2\mf Ct^2(\mf D_t^{u,v}+1)^{-\be}+\frac{32\mf m^4t^4}{1-2\mf m t}
\leq2\mf Ct^2(\mf D_t^{u,v}+1)^{-\be}+64\mf m^4t^4.
\end{multline*}

\subsubsection{Proof of \eqref{Equation: Ranges Tail Bound}}
\label{Section: Covariance Decay Step 2}

Let us denote by $\mc S_t(X)$ the number of jumps that $X$ makes
in the time interval $[0,t]$. For every $x>0$ and $v\in\ms V$, it is easy to see that
\begin{align}
\label{Equation: Jumps Tail Bound 1}
\mbf P^v\left[\max_{0\leq s\leq t} \msf d\big(v,X(s)\big)\geq x\right]
\leq\mbf P^v\big[\mc S_t(X)\geq x\big].
\end{align}
For every $v\in\ms V$ and $t\geq0$,
the number of jumps $\mc S_t(X)$ is stochastically dominated by
a poisson random variable with parameter $t\mf q$.
Therefore, applying the Chernoff bound for the tails of Poisson
random variables, we obtain that
\begin{align}
\label{Equation: Tail Bound}
\sup_{v\in\ms V}\mbf P^v\left[\max_{0\leq s\leq t} \msf d\big(v,X(s)\big)\geq x\right]
\leq\sup_{v\in\ms V}\mbf P^v\big[\mc S_t(X)\geq x\big]
\leq\mr e^{-\mf q t}\left(\frac{\mf q\mr e t}{x}\right)^{x}
\end{align}
for every $x>\mf q t$.
In order to specialize this to \eqref{Equation: Ranges Tail Bound},
we use the parameter $x:=\msf d(u,v)/4$. If $t<\min\{4/\mf q,1/4\mf q\mr e\}$ and $\msf d(u,v)\geq16$,
then we have that $4\mf q\mr e t<1$ and $x>\mf qt$,
and thus it follows by a union bound that
\begin{multline*}
\mbf P\big[(E^{u,v}_t)^c\big]^{1/2}\leq\left(\mbf P^u\left[\mc S_t(X)
\geq\frac{\msf d(u,v)}{4}\right]+\mbf P^v\left[\mc S_t(X)\geq\frac{\msf d(u,v)}{4}\right]\right)^{1/2}\\
\leq \sqrt 2\mr e^{-\mf q t/2}\left(\frac{4\mf q\mr e t}{\msf d(u,v)}\right)^{\msf d(u,v)/8}\leq\frac{\sqrt 2\,\mf q^2\mr e^2 t^2}{16},
\end{multline*}
as desired.

\subsection{Proof of Lemma \ref{Lemma: Variance Upper Bound 3}}
\label{sec:PrLem3}

\begin{notation}
Throughout this proof, we use $C>0$ to denote a constant
whose exact value may change from one display to the next.
If $C>0$ depends on some other parameters, this will be explicitly
stated.
\end{notation}

\subsubsection{Step 1. General Upper Bound}

Our first step in this proof is to provide a general upper bound for $\mbf E[\mr e^{-2\langle L^u_t+\tilde L^v_t,V\rangle}]^{1/2}$
that formalizes the intuition \eqref{Equation: Lemma 3 Intuition}.
To this effect, we claim that if \eqref{Equation: V Pointwise Lower Bound} holds, then
\begin{align}
\label{Equation: Stay Same Place Argument}
-\langle L^u_t,V\rangle
\leq-\big(\ka t^{1/\al}\msf d(0,u)\big)^{\min\{\al,1\}}+\max_{0\leq s\leq t}\Big(\ka t^{1/\al}\msf d\big(u,X^u(s)\big)\Big)^{\min\{\al,1\}}-1+\mu t
\end{align}
for every $u\in\ms V$ and $t>0$, and similarly for $-\langle\tilde L^v_t,V\rangle$.
To see this, we note that
\begin{align}\label{eq:VBound}
\nonumber
-\langle L^u_t,V\rangle
&\leq-\int_0^t\Big(\ka\,\msf d\big(0,X^u(s)\big)\Big)^\al\d s+\mu t\\
\nonumber
&=-\int_0^t\Big|\ka\Big(\msf d(0,u)-\msf d(0,u)+\msf d\big(0,X^u(s)\big)\Big)\Big|^\al\d s+\mu t\\
&=-\int_0^1\Big|\ka t^{1/\al}\Big(\msf d(0,u)-\msf d(0,u)+\msf d\big(0,X^u(ut)\big)\Big)\Big|^\al\d u+\mu t,
\end{align}
where the first line follows directly from \eqref{Equation: V Pointwise Lower Bound},
and the last line follows from a change of variables. For any $x,y\in \mathbb{R}$, the triangle
inequality implies that
\[|x-y|^{\al}\geq|x-y|^{\min\{\al,1\}}-1\geq|x|^{\min\{\al,1\}}-|y|^{\min\{\al,1\}}-1.\]
Applying this to \eqref{eq:VBound} yields
\[-\langle L^u_t,V\rangle\leq-\big(\ka t^{1/\al}\msf d(0,u)\big)^{\min\{\al,1\}}+\max_{0\leq s\leq t}\Big|\ka t^{1/\al}\Big(\msf d\big(0,X^u(s)\big)-\msf d(0,u)\Big)\Big|^{\min\{\al,1\}}-1+\mu t.\]
We then obtain \eqref{Equation: Stay Same Place Argument} by combining the fact that $x\mapsto x^{\min\{\al,1\}}$ is increasing for $x>0$
with the reverse triangle inequality $\big|\msf d\big(0,X^u(s)\big)-\msf d(0,u)\big|\leq\msf d\big(u,X^u(s)\big)$.

With \eqref{Equation: Stay Same Place Argument} in hand, we see that $\mbf E[\mr e^{-2\langle L^u_t+\tilde L^v_t,V\rangle}]^{1/2}$
is bounded above by
\begin{multline}
\label{Equation: Holder in Exponential Moment}
\mr e^{2(\mu t-1)-(\ka t^{1/\al}\msf d(0,u))^{\min\{\al,1\}}-(\ka t^{1/\al}\msf d(0,v))^{\min\{\al,1\}}}\\
\cdot\mbf E\left[\exp\left(\max_{0\leq s\leq t}\Big(\ka t^{1/\al}\msf d\big(u,X^u(s)\big)\Big)^{\min\{\al,1\}}+\max_{0\leq s\leq t}\Big(\ka t^{1/\al}\msf d\big(v,\tilde X^v(s)\big)\Big)^{\min\{\al,1\}}\right)\right]^{1/2}.
\end{multline}
On the one hand, $\mr e^{2(\mu t-1)}\to\mr e^{-2}$ as $t\to0$ for any choice of $\mu>0$. On the other hand,
thanks to the tail bound \eqref{Equation: Tail Bound}, we know that for every $\theta,\ka>0$, one has
\[\limsup_{t\to0}\sup_{u\in\ms V}\mbf E\left[\exp\left(\theta\max_{0\leq s\leq t}\Big(\ka t^{1/\al}\msf d\big(u,X^u(s)\big)\Big)^{\min\{\al,1\}}\right)\right]=1,\]
and similarly for $\tilde X$. Therefore, by a straightforward application of H\"older's inequality
on the second line of \eqref{Equation: Holder in Exponential Moment}, in order to prove Lemma \ref{Lemma: Variance Upper Bound 3}, it suffices to prove that
there exists a constant $C>0$ (which only depends on $\al$, $\be$, $d$, and $\mf c$) such that
\begin{align}
\label{Equation: Variance Upper Bound 3.1 in Lemma}
&\limsup_{t\to0}t^{2d/\al}\sum_{u,v\in\ms V}
\mr e^{-(\ka t^{1/\al}\msf d(0,u))^{\min\{\al,1\}}-(\ka t^{1/\al}\msf d(0,v))^{\min\{\al,1\}}}
\leq C\ka^{-2d};\\
\label{Equation: Variance Upper Bound 3.2 in Lemma}
&\limsup_{t\to0}t^{(2d-\be)/\al}\sum_{u,v\in\ms V}
\frac{\mr e^{-(\ka t^{1/\al}\msf d(0,u))^{\min\{\al,1\}}-(\ka t^{1/\al}\msf d(0,v))^{\min\{\al,1\}}}}{\big(\msf d(u,v)+1\big)^{\be}}
\leq C\ka^{-2d+\be}
\end{align}
for every $0<\be<d$; and
\begin{align}
\label{Equation: Variance Upper Bound 3.3 in Lemma}
&\limsup_{t\to0}t^{d/\al}\sum_{u,v\in\ms V}
\frac{\mr e^{-(\ka t^{1/\al}\msf d(0,u))^{\min\{\al,1\}}-(\ka t^{1/\al}\msf d(0,v))^{\min\{\al,1\}}}}{\big(\msf d(u,v)+1\big)^{\be}}
\leq C\ka^{-d}
\end{align}
for every $\be>d$.
We now prove these claims in two steps.

\subsubsection{Step 2. Proof of \eqref{Equation: Variance Upper Bound 3.1 in Lemma}}
\label{Section: Proof of 3.1}

Recalling the definition and upper bound of $\ms G$'s coordination sequences $\msf c_n(v)$ in
\eqref{Equation: Coordination Sequence}, we have that
\begin{align}
\label{Equation: 3.1 Proof}
\nonumber
&\sum_{u,v\in\ms V}\mr e^{-(\ka t^{1/\al}\msf d(0,u))^{\min\{\al,1\}}-(\ka t^{1/\al}\msf d(0,v))^{\min\{\al,1\}}}
=\left(\sum_{v\in\ms V}\mr e^{-(\ka t^{1/\al}\msf d(0,v))^{\min\{\al,1\}}}\right)^2\\
\nonumber
&=\left(\sum_{n\in\mbb N\cup\{0\}}\msf c_n(0)\,\mr e^{-(\ka t^{1/\al}n)^{\min\{\al,1\}}}\right)^2
\leq \mf c^2\left(\sum_{n\in\mbb N\cup\{0\}}n^{d-1}\mr e^{-(\ka t^{1/\al}n)^{\min\{\al,1\}}}\right)^2\\
&=\mf c^2t^{(-2d+2)/\al}\left(\sum_{n\in t^{1/\al}\mbb N\cup\{0\}}n^{d-1}\mr e^{-(\ka n)^{\min\{\al,1\}}}\right)^2.
\end{align}
By a Riemann sum, we have that
\begin{multline}
\label{Equation: 3.1 Proof 2}
\lim_{t\to\infty}t^{2/\al}\left(\sum_{n\in t^{1/\al}\mbb N\cup\{0\}}n^{d-1}\mr e^{-(\ka n)^{\min\{\al,1\}}}\right)^2\\
=\left(\int_0^\infty x^{d-1}\mr e^{-(\ka x)^{\min\{\al,1\}}}\d x\right)^2
=\frac{\ka^{-2 d} \Gamma \left(\frac{d}{\min\{1,\al\}}\right)^2}{\min\{1,\al^2\}}.
\end{multline}
Combining this limit with \eqref{Equation: 3.1 Proof}
yields \eqref{Equation: Variance Upper Bound 3.1 in Lemma}, where, as shown
on the right-hand side of \eqref{Equation: 3.1 Proof 2}, the constant $C>0$
only depends on the parameters $\al$, $d$, and $\mf c$.

\subsubsection{Step 3. Proof of \eqref{Equation: Variance Upper Bound 3.2 in Lemma} and \eqref{Equation: Variance Upper Bound 3.3 in Lemma}}

We now conclude the proof of Lemma \ref{Lemma: Variance Upper Bound 3} by establishing
\eqref{Equation: Variance Upper Bound 3.2 in Lemma} and \eqref{Equation: Variance Upper Bound 3.3 in Lemma}.
We separate the analysis of the sum on the left-hand sides of
\eqref{Equation: Variance Upper Bound 3.2 in Lemma} and \eqref{Equation: Variance Upper Bound 3.3 in Lemma}
into two parts, namely, the terms $u,v\in\ms V$ such that $\msf d(u,v)>\ka^{-1}t^{-1/\al}$, and those such that
$\msf d(u,v)\leq\ka^{-1}t^{-1/\al}$.

We first consider the terms such that $\msf d(u,v)>\ka^{-1}t^{-1/\al}$. For these, we have the sequence of upper bounds
\begin{align*}
&\sum_{\substack{u,v\in\ms V\\\msf d(u,v)>\ka^{-1}t^{-1/\al}}}\frac{\mr e^{-(\ka t^{1/\al}\msf d(0,u))^{\min\{\al,1\}}-(\ka t^{1/\al}\msf d(0,v))^{\min\{\al,1\}}}}{
\big(\msf d(u,v)+1\big)^{\be}}\\
&\leq\sum_{\substack{u,v\in\ms V\\\msf d(u,v)>\ka^{-1}t^{-1/\al}}}\frac{\mr e^{-(\ka t^{1/\al}\msf d(0,u))^{\min\{\al,1\}}-(\ka t^{1/\al}\msf d(0,v))^{\min\{\al,1\}}}}{
\msf d(u,v)^{\be}}\\
&<\ka^\be t^{\be/\al}\sum_{\substack{u,v\in\ms V\\\msf d(u,v)>\ka^{-1}t^{-1/\al}}}\mr e^{-(\ka t^{1/\al}\msf d(0,u))^{\min\{\al,1\}}-(\ka t^{1/\al}\msf d(0,v))^{\min\{\al,1\}}}\\
&\leq\ka^\be t^{\be/\al}\left(\sum_{v\in\ms V}\mr e^{-(\ka t^{1/\al}\msf d(0,v))^{\min\{\al,1\}}}\right)^2.
\end{align*}
At this point, by replicating the arguments in Section \ref{Section: Proof of 3.1}, we get that there exists
a constant $C>0$ that only depends on $\al$, $d$, and $\mf c$, and such that
\begin{align}
\label{Equation: 3.2 Proof 1}
\limsup_{t\to0}t^{(2d-\be)/\al}\sum_{\substack{u,v\in\ms V\\\msf d(u,v)>\ka^{-1}t^{-1/\al}}}\frac{\mr e^{-(\ka t^{1/\al}\msf d(0,u))^{\min\{\al,1\}}-(\ka t^{1/\al}\msf d(0,v))^{\min\{\al,1\}}}}{
\big(\msf d(u,v)+1\big)^{\be}}\leq C\ka^{-2d+\be}
\end{align}
if $0<\be<d$; and
\begin{align}
\label{Equation: 3.3 Proof 1}
\lim_{t\to0}t^{d/\al}\sum_{\substack{u,v\in\ms V\\\msf d(u,v)>\ka^{-1}t^{-1/\al}}}\frac{\mr e^{-(\ka t^{1/\al}\msf d(0,u))^{\min\{\al,1\}}-(\ka t^{1/\al}\msf d(0,v))^{\min\{\al,1\}}}}{
\big(\msf d(u,v)+1\big)^{\be}}=0
\end{align}
if $\be>d$.

We now consider the terms such that $\msf d(u,v)\leq\ka^{-1}t^{-1/\al}$. For those terms,
we can reformulate the summands as follows:
\begin{align}
\label{Equation: 3.2 Proof 2}
&\sum_{\substack{u,v\in\ms V\\\msf d(u,v)\leq \ka^{-1}t^{-1/\al}}}\frac{\mr e^{-(\ka t^{1/\al}\msf d(0,u))^{\min\{\al,1\}}-(\ka t^{1/\al}\msf d(0,v))^{\min\{\al,1\}}}}{
\big(\msf d(u,v)+1\big)^{\be}}\\
\nonumber
&=\sum_{u\in\ms V}\mr e^{-(\ka t^{1/\al}\msf d(0,u))^{\min\{\al,1\}}}
\left(\sum_{\substack{v\in\ms V\\\msf d(u,v)\leq \ka^{-1}t^{-1/\al}}}\frac{\mr e^{-(\ka t^{1/\al}\msf d(0,v))^{\min\{\al,1\}}}}{
\big(\msf d(u,v)+1\big)^{\be}}\right)\\
\nonumber
&=\sum_{u\in\ms V}\mr e^{-(\ka t^{1/\al}\msf d(0,u))^{\min\{\al,1\}}}
\left(\sum_{\substack{v\in\ms V\\\msf d(u,v)\leq \ka^{-1}t^{-1/\al}}}\frac{\mr e^{-(\ka t^{1/\al}(\msf d(u,v)+\msf d(0,v)-\msf d(u,v)))^{\min\{\al,1\}}}}{
\big(\msf d(u,v)+1\big)^{\be}}\right).
\end{align}
For every every $u,v\in\ms V$ such that $\msf d(u,v)\leq \ka^{-1}t^{-1/\al}$, the fact that $\msf d(0,v)\geq0$ gives the upper bound
$\mr e^{-(\ka t^{1/\al}(\msf d(0,v)-\msf d(u,v)))^{\min\{\al,1\}}}\leq\mr e$. Putting this into the above equation, we then obtain that
\begin{align*}
\eqref{Equation: 3.2 Proof 2}&\leq\mr e\sum_{u\in\ms V}\mr e^{-(\ka t^{1/\al}\msf d(0,u))^{\min\{\al,1\}}}
\left(\sum_{\substack{v\in\ms V\\\msf d(u,v)\leq \ka^{-1}t^{-1/\al}}}\frac{\mr e^{-(\ka t^{1/\al}\msf d(u,v))^{\min\{\al,1\}}}}{
\big(\msf d(u,v)+1\big)^{\be}}\right)\\
&\leq\mr e\sum_{u\in\ms V}\mr e^{-(\ka t^{1/\al}\msf d(0,u))^{\min\{\al,1\}}}
\left(\sum_{n=0}^{\ka^{-1}t^{-1/\al}}\frac{\msf c_n(u)\,\mr e^{-(\ka t^{1/\al}n)^{\min\{\al,1\}}}}{
\big(n+1\big)^{\be}}\right).
\end{align*}
Thanks to the uniform bound in \eqref{Equation: Coordination Sequence},
we then have that
\begin{align}
\label{Equation: 3.2 Proof 3}
\nonumber
\eqref{Equation: 3.2 Proof 2}&\leq\mr e \mf c\,\left(\sum_{u\in\ms V}\mr e^{-(\ka t^{1/\al}\msf d(0,u))^{\min\{\al,1\}}}\right)
\left(\sum_{n=0}^{\ka^{-1}t^{-1/\al}}\frac{n^{d-1}\mr e^{-(\ka t^{1/\al}n)^{\min\{\al,1\}}}}{
\big(n+1\big)^{\be}}\right)\\
\nonumber
&\leq \mr e^{1+(\ka t^{1/\al})^{\min\{\al,1\}}} \mf c\left(\sum_{u\in\ms V}\mr e^{-(\ka t^{1/\al}\msf d(0,u))^{\min\{\al,1\}}}\right)\\
&\nonumber\hspace{1.5in}
\cdot\left(\sum_{n\in\mbb N\cup\{0\}}(n+1)^{d-1-\be}\mr e^{-(\ka t^{1/\al}(n+1))^{\min\{\al,1\}}}\right)\\
&=\mr e^{1+o(1)} \mf c\,\left(\sum_{u\in\ms V}\mr e^{-(\ka t^{1/\al}\msf d(0,u))^{\min\{\al,1\}}}\right)
\left(\sum_{n\in \mbb N}n^{d-1-\be}\mr e^{-(\ka t^{1/\al} n)^{\min\{\al,1\}}}\right).
\end{align}

We now analyze the two sums on the right-hand side of \eqref{Equation: 3.2 Proof 3}.
Looking at the first term, the same analysis carried out in Section \ref{Section: Proof of 3.1} implies that
\begin{align*}
\label{Equation: 3.x Proof 3}
\limsup_{t\to0}t^{d/\al}\sum_{u\in\ms V}\mr e^{-(\ka t^{1/\al}\msf d(0,u))^{\min\{\al,1\}}}\leq C\ka^{-d}
\end{align*}
for some $C$ that only depends on $\al$, $d$, and $\mf c$.
Next, the second sum in \eqref{Equation: 3.2 Proof 3} is analyzed differently depending on whether $0<\be<d$ or $\be>d$:
On the one hand, if $\be<d$, then by a Riemann sum we have that
\begin{multline*}
\lim_{t\to0}t^{(d-\be)/\al}\sum_{n\in \mbb N}n^{d-1-\be}\mr e^{-(\ka t^{1/\al} n)^{\min\{\al,1\}}}
=\lim_{t\to0}t^{1/\al}\sum_{n\in t^{1/\al}\mbb N}n^{d-1-\be}\mr e^{-(\ka n)^{\min\{\al,1\}}}\\
=\int_0^\infty x^{d-1-\be}\mr e^{-(\ka x)^{\min\{\al,1\}}}\d x
=\frac{\kappa ^{-d+\be} \Gamma \left(\frac{d-\beta }{\min\{\al,1\}}\right)}{\min\{\al,1\}}.
\end{multline*}
On the other hand, if $\be>d$, then we have by dominated convergence that
\[\lim_{t\to0}\sum_{n\in \mbb N}n^{d-1-\be}\mr e^{-(\ka t^{1/\al} n)^{\min\{\al,1\}}}
=\sum_{n\in \mbb N}n^{d-1-\be};\]
we know that the sum on the right-hand side is convergent since $\be>d$.

Putting these two limits back into \eqref{Equation: 3.2 Proof 3}, we then get that there exists a constant $C>0$
(which only depends on $\al$, $d$, $\be$, and $\mf c$) such that
\[\limsup_{t\to0}t^{(2d-\be)/\al}\sum_{\substack{u,v\in\ms V\\\msf d(u,v)\leq \ka^{-1}t^{-1/\al}}}\frac{\mr e^{-(\ka t^{1/\al}\msf d(0,u))^{\min\{\al,1\}}-(\ka t^{1/\al}\msf d(0,v))^{\min\{\al,1\}}}}{
\big(\msf d(u,v)+1\big)^{\be}}\leq C\ka^{-2d+\be}\]
when $\be<d$, and such that
\[\limsup_{t\to0}t^{d/\al}\sum_{\substack{u,v\in\ms V\\\msf d(u,v)\leq \ka^{-1}t^{-1/\al}}}\frac{\mr e^{-(\ka t^{1/\al}\msf d(0,u))^{\min\{\al,1\}}-(\ka t^{1/\al}\msf d(0,v))^{\min\{\al,1\}}}}{
\big(\msf d(u,v)+1\big)^{\be}}\leq C\ka^{-d}\]
when $\be>d$.
Combining this with \eqref{Equation: 3.2 Proof 1} and \eqref{Equation: 3.3 Proof 1} concludes the proof of \eqref{Equation: Variance Upper Bound 3.2 in Lemma}
and \eqref{Equation: Variance Upper Bound 3.3 in Lemma}.
With this in hand, we have now completed the proof of Lemma \ref{Lemma: Variance Upper Bound 3}.

\section{Spectral Mapping and Multiplicity}
\label{sec: Multiplicity}

A crucial aspect of the proof of Theorem \ref{Theorem: Rigidity}
is the ability to relate exponential linear statistics of the eigenvalue point process
\eqref{Equation: Eigenvalue Point Process} to the trace of $K_t$ via the identities
\begin{align}
\label{Equation: Trace Identity}
\mr{Tr}[K_t]=\sum_{\mu\in\si(K_t)\setminus\{0\}}m_a(\mu,K_t)\,\mu=\sum_{\la\in\si(H)}m_a(\la,H)\,\mr e^{-t\la}\in(0,\infty).
\end{align}
Though we expect that such a result is known (or at least folklore) in the operator theory
community, we were not able to locate any reference that contains all of the precise statements
that we need to prove \eqref{Equation: Trace Identity}.
(This is especially so since the level of generality in this paper allows for non-self-adjoint
operators.)
As such, our purpose in this section
is to provide a general criterion for an identity of the form \eqref{Equation: Trace Identity}
to hold (as well as a few more properties), which we then use in Section \ref{Section: Rigidity}
to wrap up the proof of Theorem \ref{Theorem: Rigidity}.

We begin this section with a definition:

\begin{definition}
\label{Definition: Finite Dimensional}
We say that a linear operator $T$ on $\ell^2_\ms Z(\ms V)$
is finite-dimensional if there exists a finite set $\ms U\subset\ms V$
such that $T(u,v)=0$ whenever $(u,v)\not\in\ms U\times\ms U$.
In particular, if we enumerate the set $\ms U=\{u_1,\ldots,u_{|\ms U|}\}$, then $T$
has the same spectrum as the $|\ms U|\times |\ms U|$ matrix $M_T$ with entries
\begin{align}
\label{Equation: Matrix Representation}
M_T(i,j):=T(u_i,u_j),\qquad 1\leq i,j\leq |\ms U|.
\end{align}
\end{definition}

The result that we prove in this section is as follows:

\begin{proposition}
\label{Proposition: Operator Theory}
Let $(T_t)_{t>0}$ be a strongly continuous semigroup of trace class operators on $\ell^2_\ms Z(\ms V)$
such that $\|T_t\|_{\mr{op}}\leq\mr e^{-\om t}$ for some $\om<0$, and
let $G$ be its infinitesimal generator.
The following holds:
\begin{enumerate}
\item $G$ is closed and densely defined on $\ell^2_\ms Z(\ms V)$.
\item $\si(G)=\si_p(G)$, and $\Re(\la)\geq\om$ for all $\la\in\si(G)$.
\item For every $t>0$, $\si(T_t)\setminus\{0\}=\{\mr e^{-t\la}:\la\in\si(G)\}$.
\end{enumerate}
Moreover, if there exists a sequence of finite-dimensional
operators $(G_n)_{n\in\mbb N}$ such that
\begin{align}
\label{Equation: Resolvent Convergence Assumption}
\lim_{n\to\infty}\|\mf R(z,G_n)-\mf R(z,G)\|_{\mr{op}}=0
\end{align}
for at least one $z\in\mbb C\setminus\si(G)$ and such that
\begin{align}
\label{Equation: Semigroup Convergence Assumption}
\lim_{n\to\infty}\|\mr e^{-t G_n}-T_t\|_{\mr{op}}=0,
\end{align}
then for every $t>0$ and $\mu\in\si(T_t)\setminus\{0\}$,
\begin{align}
\label{Equation: Multiplicity Identity}
m_a(\mu,T_t)=\sum_{\la\in\si(G):~\mr e^{-t\la}=\mu}m_a(\la,G).
\end{align}
\end{proposition}

As a direct consequence of the above proposition, we have that
\[\mr{Tr}[T_t]=\sum_{\mu\in\si(T_t)\setminus\{0\}}m_a(\mu,T_t)\,\mu=\sum_{\la\in\si(G)}m_a(\la,G)\,\mr e^{-t\la}\in\mbb C\]
for all $t>0$, which is precisely the kind of statement that we are looking for.
The remainder of this section is now devoted to the proof of Proposition \ref{Proposition: Operator Theory}.

\subsection{Step 1. Closed Generator and Spectral Mapping}

We begin with the more straightforward aspects of the statement
of Proposition \ref{Proposition: Operator Theory}, namely, items (1)--(3).
Since $(T_t)_{t>0}$ is strongly continuous and $\|T_t\|_{\mr{op}}\leq\mr e^{-\om t}$,
it follows from the Hille-Yosida theorem (e.g., \cite[Chapter II, Corollary 3.6]{EngelNagel})
that $G$ is closed and densely defined on $\ell^2_\ms Z(\ms V)$.
Moreover, $\Re(\la)\geq\om$ for every $\la\in\si(G)$.
Given that the $T_t$ are trace class, we know that $\si(T_t)=\si_p(T_t)$ and that
\[\mr{Tr}[T_t]=\sum_{\mu\in\si(T_t)\setminus\{0\}}m_a(\mu,T_t)\mu\in\mbb C\]
by Lidskii's theorem
(e.g., \cite[Sections 3.6 and 3.12]{Simon}). Next,
by the spectral mapping theorem (e.g., \cite[Chapter IV, (3.7) and (3.16)]{EngelNagel}),
we know that for every $t>0$,
\begin{align}
\label{Equation: Spectral Mapping}
\big\{\mr e^{-t\la}:\la\in\si(G)\big\}\subset\si(T_t)
\qquad\text{and}\qquad
\big\{\mr e^{-t\la}:\la\in\si_p(G)\big\}=\si_p(T_t)\setminus\{0\}.
\end{align}
In particular, $\si(G)=\si_p(G)$, concluding the proof of 
Proposition \ref{Proposition: Operator Theory} (1)--(3).

\subsection{Step 2. Multiplicities in Finite Dimensions}

It now remains to prove \eqref{Equation: Multiplicity Identity}.
Before we prove this result, we first prove the corresponding statement
in finite dimensions, namely:

\begin{lemma}
\label{Lemma: Spectral Mapping in Finite Dimensions}
Let $T$ be a finite-dimensional linear operator on $\ell^2_\ms Z(\ms V)$
and $F:\mbb C\to\mbb C$ be an analytic function.
For every $\mu\in\si\big(F(T)\big)=F\big(\si(T)\big)$, one has
\[m_a\big(\mu,F(T)\big)=\sum_{\la\in\si(T):~F(\la)=\mu}m_a(\la,T).\]
\end{lemma}

Applying this to the exponential map and the operators $G_n$, we are led to
the fact that for every $n\in\mbb N$, $t>0$, and $\mu\in\si(\mr e^{-tG_n})$ one has
\begin{align}
\label{Equation: Algebraic Identity Prelimit}
m_a(\mu,\mr e^{-t G_n})=\sum_{\la\in\si(G_n):~\mr e^{-t\la}=\mu}m_a(\la,G_n).
\end{align}

\begin{proof}[Proof of Lemma \ref{Lemma: Spectral Mapping in Finite Dimensions}]
It suffices to prove the result with $T$ replaced by $M_T$
and $F(T)$ replaced by $F(M_T)$, where $M_T$ is the
matrix defined in \eqref{Equation: Matrix Representation}.
Let $M_T=PJP^{-1}$ be $M_T$'s Jordan canonical form. That is,
$J$ is the direct sum of $M_T$'s Jordan blocks, and in particular the number of times
any $\la\in\mbb C$ appears on $J$'s diagonal is equal to
$m_a(\la,M_T)$. By the standard analytic functional calculus for matrices,
we know that $F(M_T)=PF(J)P^{-1}$, where $F(J)$ is the direct sum
of $M_T$'s transformed Jordan blocks, wherein any $k\times k$ Jordan block
of the form
\[\left[\begin{array}{ccccc}\la&1\\
&\la&1\\
&&\ddots&\ddots\\
&&&\la&1\end{array}\right]\]
is transformed into the upper triangular matrix
\[\left[\begin{array}{ccccc}F(\la)&F'(\la)&F''(\la)/2&\cdots&F^{(k-1)}(\la)/(k-1)!\\
&F(\la)&F'(\la)&\cdots&F^{(k-2)}(\la)/(k-2)!\\
&&\ddots&\ddots&\vdots\\
&&&\ddots&F'(\la)\\
&&&&F(\la)\end{array}\right].\]
Given that the characteristic polynomial of $F(M_T)$ is the same as that of $F(J)$,
this readily implies the result.
\end{proof}

\subsection{Step 3. Passing to the Limit}

We now complete the proof of Proposition \ref{Proposition: Operator Theory}
by arguing that the identity \eqref{Equation: Algebraic Identity Prelimit} persists
in the large $n$ limit.
Thanks to \eqref{Equation: Resolvent Convergence Assumption}
and \eqref{Equation: Semigroup Convergence Assumption}, we know that we have the convergences
$G_n\to G$ and $\mr e^{-t G_n}\to T_t$ for every $t>0$ in the generalized sense of Kato
(see \cite[Chapter IV, (2.9), (2.20) and p. 206]{Kato} for a definition of convergence in the
generalized sense, and \cite[Chapter IV, Theorems 2.23 a) and 2.25]{Kato} for a proof
that norm-resolvent and norm convergence implies convergence in the generalized sense).
As shown in \cite[Chapter IV, Theorem 3.16]{Kato} (see also \cite[Chapter IV, Section 5]{Kato}
for a discussion specific to the context of isolated eigenvalues),
convergence in the generalized sense implies the following spectral continuity results:

\begin{notation}
In what follows, we use $B(z,r)$ to denote the closed ball in the complex plane centered at $z\in\mbb Z$
and with raduis $r>0$.
\end{notation}

\begin{corollary}
For every $\la\in\si(G)$, if $\eps>0$ is such that $\si(G)\cap B(\la,\eps)=\{\la\}$,
then there exists $N\in\mbb N$ large enough so that
\begin{align}
\label{Equation: Multiplicities Convergence 1}
\sum_{\tilde\la\in\si(G_n)\cap B(\la,\eps)}m_a(\tilde\la,G_n)=m_a(\la,G)
\end{align}
whenever $n\geq N$.

Conversely, for every $t>0$ and $\mu\in\si(T_t)\setminus\{0\}$, if $\eps>0$ is such that $\si(T_t)\cap B(\mu,\eps)=\{\mu\}$,
then there exists $N\in\mbb N$ large enough so that
\begin{align}
\label{Equation: Multiplicities Convergence 2}
\sum_{\tilde\mu\in(\mr e^{-t G_n})\cap B(\mu,\eps)}m_a(\tilde\mu,\mr e^{-t G_n})=m_a(\mu,T_t)
\end{align}
whenever $n\geq N$.
\end{corollary}

We are now ready to prove \eqref{Equation: Matrix Representation}.
We first show that for every $t>0$ and $\mu\in\si(T_t)\setminus\{0\}$, the set $\{\la\in\si(G):~\mr e^{-t\la}=\mu\}$
is finite. Suppose by contradiction that this is not the case. Then, for any integer $M>0$, we can
find at least $M$ distinct eigenvalues $\la_1,\ldots,\la_M\in\si(G)$ such that $\mr e^{-t\la_i}=\mu$.
By taking a small enough $\eps>0$ and large enough $N\in\mbb N$, a combination of
\eqref{Equation: Algebraic Identity Prelimit} and \eqref{Equation: Multiplicities Convergence 2} yields
\begin{align}
\label{Equation: Algebraic Identity Prelimit Finite 1}
m_a(\mu,T_t)=\sum_{\tilde\mu\in\si(\mr e^{-t G_N})\cap B(\mu,\eps)}m_a(\tilde\mu,\mr e^{-tG_N})
=\sum_{\tilde\la\in\si(G_N):~\mr e^{-t\tilde\la}\in B(\mu,\eps)}m_a(\tilde\la,G_N).
\end{align}
Since $z\mapsto\mr e^{-t z}$ is continuous, we can take $\de>0$ small enough so that
\begin{enumerate}
\item if $\tilde\la\in B(\la_i,\de)$ for some $1\leq i\leq M$, then $\mr e^{-t\tilde\la}\in B(\mu,\eps)$; and
\item $\si(G)\cap B(\la_i,\de)=\{\la_i\}$ for every $1\leq i\leq M$.
\end{enumerate}
Thus, up to increasing the value of $N$ if necessary, an application of \eqref{Equation: Multiplicities Convergence 1}
to the right-hand side of \eqref{Equation: Algebraic Identity Prelimit Finite 1} then gives
\begin{align}
\label{Equation: Algebraic Identity Prelimit Finite 2}
m_a(\mu,T_t)\geq\sum_{i=1}^{M}\sum_{\tilde\la\in\si(G_N)\cap B(\la_i,\de)}m_a(\tilde\la,G_N)
=\sum_{i=1}^Mm_a(\la_i,G)\geq M.
\end{align}
Since $M$ was arbitrary, this implies that $m_a(\mu,T_t)=\infty$. Since $T_t$ is trace class
this cannot be the case, hence
we conclude that $\{\la\in\si(G):~\mr e^{-t\la}=\mu\}$ is finite.

By repeating the argument leading up to \eqref{Equation: Algebraic Identity Prelimit Finite 2},
but this time letting $M$ be equal to the number of eigenvalues in the set
$\{\la\in\si(G):~\mr e^{-t\la}=\mu\}$, we obtain that
\[m_a(\mu,T_t)\geq\sum_{\la\in\si(G):~\mr e^{-t\la}=\mu}m_a(\la,G).\]
We now proceed to prove the reverse inequality. Recall that $\{\la\in\si(G):~\mr e^{-t\la}=\mu\}$ contains finitely many elements. Denote them by $\lambda_1, \ldots , \lambda_M$ for some $M\in \mathbb{N}$.
Thanks to \eqref{Equation: Multiplicities Convergence 1}, we can find
a small enough $\eps>0$ and large enough $N\in\mbb N$ such that
\begin{align*}
 \sum_{i=1}^{M} m_{a}(\lambda_i, G) &=\sum_{\tilde{\lambda} \in \cup^{M}_{i=1}\sigma(G_N)\cap B(\lambda_i, \eps)}m_{a}(\tilde{\lambda}, G_N) = \sum_{\tilde{\lambda} \in \sigma(G_N)\cap \big(\cup^{M}_{i=1} B(\lambda_i, \eps)\big)}m_{a}(\tilde{\lambda}, G_N).
\end{align*}
Then, by \eqref{Equation: Algebraic Identity Prelimit}, one has
\begin{align}
\sum_{\tilde{\lambda} \in \sigma(G_N)\cap \big(\cup^{M}_{i=1} B(\lambda_i, \eps)\big)}m_{a}(\tilde{\lambda}, G_N) = \sum_{\substack{\tilde{\mu}\in \sigma(e^{-tG_N})\\\tilde{\mu}\in e^{-t}(\cup_{i=1}^{M}B(\lambda_i,\eps))}}m(\tilde{\mu}, e^{-tG_N}),\label{eq:RevIneq2}
\end{align}
where we use $\mr e^{-t}(B)$ to denote the image of a set $B\subset\mbb C$ through the exponential map $z\mapsto\mr e^{-tz}$.
Since the exponential map is open and $\mr e^{-t\la_i}=\mu$ for all $1\leq i\leq M$, we can find a small enough $\de>0$
such that $B(\mu,\delta)\subset e^{-t}(\cup_{i=1}^{M}B(\lambda_i,\eps))$ and $\si(T_t)\cap B(\mu,\de)=\{\mu\}$.
As a result we get 
\begin{align}
 \sum_{i=1}^{M} m_{a}(\lambda_i, G)\geq\text{r.h.s. of \eqref{eq:RevIneq2}}\geq \sum_{\tilde{\mu}\in \sigma(e^{-tG_N})\cap B(\mu,\delta)} m_a(\tilde{\mu}, e^{-tG_N}).
\end{align}
At this point, up to increasing $N$ if necessary an application of
\eqref{Equation: Multiplicities Convergence 2} then yields
\[ \sum_{i=1}^{M} m_{a}(\lambda_i, G)\geq\sum_{\tilde{\mu}\in \sigma(e^{-tG_N})\cap B(\mu,\delta)} m_a(\tilde{\mu}, e^{-tG_N})
=m_a(\mu,T_t),\]
thus concluding the proof of \eqref{Equation: Multiplicity Identity}
and Proposition \ref{Proposition: Operator Theory}.

\section{Proof of Theorem \ref{Theorem: Rigidity}}
\label{Section: Rigidity}

In this section, we prove Theorem~\ref{Theorem: Rigidity}.
We suppose throughout that Assumptions~\ref{Assumption: Graph} and~\ref{Assumption: Potential and Noise} hold.
We begin with a notation:

\begin{notation}
Throughout this proof,
we denote $X$'s transition semigroup by
\[\Pi_t(u,v)=\mbf P^u[X(t)=v],\qquad t\geq0,~u,v\in\ms V.\]
\end{notation}

\subsection{Step 1. Boundedness}
\label{Section: Boundedness}

Our first step in the proof is to show that, almost surely, $K_t$ is a bounded linear operator
on $\ell^2_\ms Z(\ms V)$ with $\|K_t\|_{\mr{op}}\leq\mr e^{\om t}$
for every $t>0$ for some $\om<0$. As is typical in Schr\"odinger semigroup theory,
this relies on controlling the minimum of the random potential $V+\xi$. To this end, we have the following
result:

\begin{lemma}
\label{Lemma: Bounded Below V Plus xi}
Define the random variable
\begin{align}
\label{Equation: Omega Zero}
\om_0:=\inf_{v\in\ms V}\big(V(v)+\xi(v)\big).
\end{align}
$\om_0>-\infty$ almost surely.
\end{lemma}
\begin{proof}
Thanks to \eqref{Equation: Potential Growth}, it suffices to prove that
\begin{align}
\label{Equation: Log n Borel-Cantelli Bound}
\liminf_{n\to\infty}\left(\inf_{v\in\ms V:~\msf d(0,v)\leq n}\frac{\xi(v)}{\log n}\right)>-\infty\qquad\text{almost surely}.
\end{align}
By a union bound and Markov's inequality, for every $\theta,\la>0$,
\[\mbf P\Big(\inf_{v\in\ms V:~\msf d(0,v)\leq n}\xi(v)\leq-\la\Big)
\leq\sum_{v\in\ms V:~\msf d(0,v)\leq n}\mr e^{-\theta\la}\mbf E\big[\mr e^{-\theta\xi(v)}\big].\]
On the one hand, thanks to \eqref{Equation: Coordination Sequence}, we have that
\[|\{v\in\ms V:\msf d(0,v)\leq n\}|\leq\mf c\sum_{m=1}^n m^{d-1}\leq\mf c+\mf c\int_1^nx^{d-1}\d x\leq Cn^d\]
for some constant $C>0$. On the other hand, thanks to the moment bound \eqref{Equation: Exponential Moments},
there exists a $\theta>0$ small enough so that
\[\sup_{v\in\ms V}\mbf E\big[\mr e^{-\theta\xi(v)}\big]<\infty.\]
Combining these two observations, we conclude that there exists $\tilde C,\theta>0$ such that
\[\mbf P\Big(\inf_{v\in\ms V:~\msf d(0,v)\leq n}\xi(v)\leq-\la\Big)\leq \tilde C n^d\mr e^{-\theta\la},\qquad\la>0.\]
If we take $\la=\la(n)=c\log n$ for large enough $c>0$,
then $\sum_{n\in\mbb N}\tilde C n^d\mr e^{-\theta\la(n)}<\infty$;
hence \eqref{Equation: Log n Borel-Cantelli Bound} holds by the
Borel-Cantelli lemma.
\end{proof}

As a direct application of Lemma \ref{Lemma: Bounded Below V Plus xi},
we have the inequality $K_t(u,v)\leq\mr e^{-\om_0t}\Pi_t(u,v)$ for every $u,v\in\ms V$,
where we take $\om_0$ as in \eqref{Equation: Omega Zero}.
In particular, $\|K_t\|_{\mr{op}}\leq\mr e^{-\om_0t}\|\Pi_t\|_{\mr{op}}$.
Given that $\om_0>-\infty$ almost surely by Lemma \ref{Lemma: Bounded Below V Plus xi},
it suffices to prove that $\Pi_t$ is bounded with $\|\Pi_t\|_{\mr{op}}\leq \mr e^{-t\om_1}$
for some constant $\om_1\leq 0$. We now prove this.

Note that for every $f\in\ell^2(\ms V)$, we have by Jensen's inequality that
\[\|\Pi_tf\|_{2}^2=\sum_{v\in\ms V}\mbf E^v\big[f\big(X(t)\big)\big]^2\leq\sum_{v\in\ms V}\mbf E^v\big[f\big(X(t)\big)^2\big]
=\sum_{u,v\in\ms V}\Pi_t(v,u)f(u)^2,\]
from which we conclude that
\[\|\Pi_t\|_{\mr{op}}\leq\sqrt{\sup_{u\in\ms V}\sum_{v\in\ms V}\Pi_t(v,u)}.\]
If we define the matrix
\[H_X(u,v):=\begin{cases}
-q(u)\Pi(u,v)&\text{if }u\neq v\\
q(u)&\text{if }u=v
\end{cases},\qquad u,v\in\ms V\]
(i.e., the Markov generator of $X$), then we can write
\[\sum_{v\in\ms V}\Pi_t(v,u)=\sum_{v\in\ms V}\sum_{n=0}^\infty\frac{(-t)^nH_X^n(v,u)}{n!}
\leq\sum_{n=0}^\infty\frac{t^n}{n!}\sum_{v\in\ms V}|H_X^n(v,u)|.\]
Noting that
\[\sup_{u,v\in\ms V}|H^n_X(u,v)|\leq\|H^n_X\|_{\mr{op}}\leq\|H_X\|^n_{\mr{op}},\]
for every $u,v\in\ms V$, we have the bound
\[|H_X^n(v,u)|\leq \|H_X\|^n_{\mr{op}}\mbf 1_{\{\msf d(u,v)\leq n\}}.\]
By \eqref{Equation: Coordination Sequence}, for any $u\in\ms V$,
the number of $v\in\ms V$ such that $(u,v)$ is an edge is bounded by $\mf c$.
Thus, the number of $v\in\ms V$ such that $\msf d(u,v)\leq n$ is crudely bounded by $\mf c^n$.
Consequently,
\[\|\Pi_t\|_{\mr{op}}^2\leq\sup_{u\in\ms V}\sum_{v\in\ms V}\Pi_t(v,u)\leq \sum_{n=0}^\infty\frac{(t\mf c\|H_X\|_{\mr{op}})^n}{n!}=\mr e^{\mf c\|H_X\|_{\mr{op}}t}.\]
Thus, it now suffices to prove that $\|H_X\|_{\mr{op}}<\infty$.

Recall that, by assumption, $\mf q:= \mathrm{sup}_{u\in \ms V} q(u)<\infty$. For every $f\in\ell^2(\ms V)$,
\[\|H_Xf\|_{2}^2\leq\mf q^2\sum_{u\in\ms V}\left(\sum_{v\in\ms V}\mbf 1_{\{(u,v)\in\ms E\}}f(v)\right)^2
\leq\mf q^22^{\mf c}\sum_{u,v\in\ms V}\mbf 1_{\{(u,v)\in\ms E\}}f(v)^2,\]
where the last inequality comes from the fact that
\[(x_1+\cdots+x_{\mf c})^2\leq2^{\mf c}(x_1^2+\cdots+x_{\mf c}^2),\qquad x_i\in\mbb R,\]
and that, by \eqref{Equation: Coordination Sequence},
for every $v\in\ms V$ there are at most $\mf c$ vertices $u$
such that $(u,v)\in\ms E$.
Using once again this last observation, we have that
\[\sum_{u,v\in\ms V}\mbf 1_{\{(u,v)\in\ms E\}}f(v)^2\leq\mf c\|f\|_2^2,\]
from which we conclude that
$\|H_X\|_{\mr{op}}^2\leq\mf q^22^{\mf c}\mf c,$
as desired.

\subsection{Step 2. Continuity of the Semigroup}

We now prove the almost-sure strong continuity and semigroup property.
Since $X$ is Markov and local time is additive, the semigroup property is trivial. We now prove strong continuity.
Let $C_{0,\ms Z}(\ms V)$ denote the set of functions $f:\ms V\to\mbb R$ that are finitely supported
on $\ms V\setminus\ms Z$.
Since $C_{0,\ms Z}(\ms V)$ is dense in $\ell^2_\ms Z(\ms V)$ and a semigroup of bounded linear
operators is strongly continuous if and only if it is weakly continuous
(e.g., \cite[Chapter I, Theorem 5.8]{EngelNagel}), it suffices to prove that
$\langle f,K_tg-g\rangle\to0$ as $t\to0$ for every $f,g\in C_{0,\ms Z}(\ms V)$.
For every $g\in C_{0,\ms Z}(\ms V)$, we know that
\[\lim_{t\to0}g\big(X(t)\big)\mr e^{-\langle L_t,V+\xi\rangle}=g\big(X(0)\big)\mbf 1_{\{X(0)\not\in\ms Z\}}=g\big(X(0)\big)\qquad\text{almost surely}.\]
By the definition of $\om_0$, it follows that $\langle L_t,V+\xi\rangle\geq \om_0 t$ which implies that
\[\big|g\big(X(t)\big)\mr e^{-\langle L_t,V+\xi\rangle}\big|\leq\|g\|_{\ell^\infty}\mr e^{-\om_0t}.\] Since the right-hand side of this inequality is independent of $X$, it follows from dominated
convergence that
\[\lim_{t\to0}K_tg(v)=\lim_{t\to0}\mbf E^v\left[g\big(X(t)\big)\mr e^{-\langle L_t,V+\xi\rangle}\right]=g(v)\qquad\text{almost surely}\]
for every $v\in\ms V$. Finally, given that for every $v\in\ms V$, we have 
\[\big|f(v)\big(K_tg(v)-g(v)\big)\big|\leq\|f\|_{\ell^\infty}\|g\|_{\ell^\infty}(\mr e^{-\om_0t}+1)\mbf 1_{\{f(v)\neq0\}},\]
which is summable in $v$ whenever $f\in C_{0,\ms Z}(\ms V)$,
we obtain $\langle f,K_tg-g\rangle\to0$ as $t\to0$ by dominated convergence.

\subsection{Step 3. Trace Class}

By the semigroup property, for every $t>0$, we can write
$K_t$ as the product $K_{t/2}K_{t/2}$. Thus, given that
the product of any two Hilbert-Schmidt operators is trace class
(e.g., \cite[Theorem 3.7.4]{Simon}),
it suffices to prove that, almost surely, $K_t$ is Hilbert-Schmidt
for all $t>0$, that is,
\[\sum_{u,v\in\ms V}K_t(u,v)^2<\infty.\]
By \eqref{Equation: Log n Borel-Cantelli Bound},
there exists finite random variables $\ka,\mu>0$ that only depend on $\xi$ such that
\[V(v)+\xi(v)\geq\big(\ka\msf d(0,v)\big)^\al-\mu,\qquad v\in\ms V\]
almost surely. Therefore, it suffices to prove the result with $K_t$ replaced by the kernel
\[\tilde K_t(u,v):=\mr e^{\mu t}\mbf E^u\left[\mr e^{-\langle L_t,(\ka\msf d(0,\cdot))^\al\rangle}\mbf 1_{\{X(t)=v\}}\right],\qquad u,v\in\ms V.\]
By Jensen's inequality,
\begin{align*}
\sum_{u,v\in\ms V}\tilde K_t(u,v)^2
&\leq\mr e^{2\mu t}\sum_{u,v\in\ms V}\mbf E^u\left[\mr e^{-2\langle L_t,(\ka\msf d(0,\cdot))^\al\rangle}\mbf 1_{\{X(t)=v\}}\right]\\
&\leq\mr e^{2\mu t}\sum_{u\in\ms V}\mbf E^u\left[\mr e^{-2\langle L_t,(\ka\msf d(0,\cdot))^\al\rangle}\right].
\end{align*}
At this point, the same argument used in \eqref{Equation: Tail Bound}, \eqref{eq:VBound}, and \eqref{Equation: Holder in Exponential Moment}
implies that there exists some finite constant $C_{\ka,t}>0$ (which depends on $\ka$ and $t$) such that
\[\sum_{u,v\in\ms V}\tilde K_t(u,v)^2\leq C_{\ka,t}\mr e^{2\mu t}\sum_{u\in\ms V}\mr e^{-2t(\ka\msf d(0,u))^\al}.\]
Then, writing the above sum as
\[\sum_{u\in\ms V}\mr e^{-2t(\ka\msf d(0,u))^\al}=\sum_{n\in\mbb N}\msf c_n(0)\mr e^{-2t(\ka n)^\al},\]
this is easily seen to be finite for all $t>0$ by \eqref{Equation: Coordination Sequence}.

\subsection{Step 4. Infinitesimal Generator}

We now prove the properties of the generator $H$,
except for number rigidity of its spectrum, which is relegated to the next (and final) step of the proof.
That $K_t$'s generator
is of the form \eqref{Equation: Schrodinger Generator} follows
from the straightforward computation that for every $u,v\in\ms V\setminus\ms Z$,
\[\lim_{t\to0}\frac{\mbf 1_{\{u=v\}}-K_t(u,v)}{t}=H(u,v)\qquad\text{almost surely}\]
(indeed, recall that by definition of the process $X$, $\Pi_t(u,v)=q(u)\Pi(u,v)t+o(t)$ as $t\to0$
whenever $u\neq v$, and that $K_t(u,v)=0$ if $u\in\ms Z$ or $v\in\ms Z$).

Almost surely, $(K_t)_{t>0}$ is a strongly continuous semigroup of trace class operators and $\|K_t\|_{\mr{op}}\leq\mr e^{-\om t}$.
Therefore, by Proposition \ref{Proposition: Operator Theory} (1)--(3), the following holds almost surely:
\begin{enumerate}
\item $H$ is closed and densely defined on $\ell^2_\ms Z(\ms V)$.
\item $\si(H)=\si_p(H)$, and $\Re(\la)\geq\om$ for all $\la\in\si(H)$.
\item For every $t>0$, $\si(K_t)\setminus\{0\}=\{\mr e^{-t\la}:\la\in\si(H)\}$.
\end{enumerate}
It now remains to establish the trace identity \eqref{Equation: Trace Identity}, which is crucial
in our proof of rigidity. The fact that $\mr{Tr}[K_t]$
is a positive real number follows from the fact that
\[\mr{Tr}[K_t]=\sum_{v\in\ms V}K_t(v,v)\]
and that $K_t(u,v)\in[0,\infty)$ for all $u,v\in\ms V$. To prove the remainder of \eqref{Equation: Trace Identity},
as per Proposition \ref{Proposition: Operator Theory}, we need to find a sequence of finite-dimensional
operators that converge to $H$ and $K_t$ in the sense of \eqref{Equation: Resolvent Convergence Assumption}
and \eqref{Equation: Semigroup Convergence Assumption}.

To this end, for every $n\in\mbb N$, let us denote the subset
\[\ms V_n:=\{v\in\ms V:\msf d(0,v)\leq n\}\subset\ms V.\]
Given that $\ms G$ has uniformly bounded degrees, this must be finite. Thus,
the operators
\[H_n(u,v):=H(u,v)\mbf 1_{\{(u,v)\in\ms V_n\}},\qquad u,v\in\ms V\]
are finite-dimensional in the sense of Definition \ref{Definition: Finite Dimensional}.
More specifically, $H_n$ is the restriction of $H$ to the set $\ms V_n$
with Dirichlet boundary on $\ms V\setminus\ms V_n$. In particular,
if for every $n\in\mbb N$ we denote the hitting time
\[\tau_n:= \inf_{t\geq0}\big\{t\geq 0: X(t)\not\in\ms V_n\big\},\]
Then $\mr e^{-tH_n}$ is the integral operator on $\ell^2_\ms Z(\ms V)$ with kernel
\begin{align}
\label{Equation: Finite Kernel}
\mr e^{-tH_n}(u,v)=\mbf E^u\left[\mr e^{-\langle L_t,V+\xi\rangle}\mbf 1_{\{X(t)=v\}}\mbf 1_{\{\tau_n>t\}}\right].
\end{align}
The proof of \eqref{Equation: Trace Identity} is now a matter of establishing the following result:

\begin{lemma}
Almost surely, it holds that
\begin{align}
\label{Equation: Resolvent Convergence Assumption 2}
\lim_{n\to\infty}\|\mf R(z,H_n)-\mf R(z,H)\|_{\mr{op}}=0
\end{align}
for every $z\in\mbb C$ such that $\Re(z)<\om$ and
\begin{align}
\label{Equation: Semigroup Convergence Assumption 2}
\lim_{n\to\infty}\|\mr e^{-t G_n}-K_t\|_{\mr{op}}=0
\end{align}
for every $t>0$.
\end{lemma}
\begin{proof}
Given that $0\leq\mr e^{-tH_n}(u,v)\leq K_t(u,v)$ for all $u,v\in\ms V$,
it is easy to see that $\|\mr e^{-tH_n}\|_{\mr{op}}\leq\|K_t\|_{\mr{op}}\leq\mr e^{-\om t}$
for all $t>0$ almost surely. In particular, any $z\in\mbb C$ such that $\Re(z)<\om$
is in the resolvent set of $H_n$ and $H$ for all $n$. Consequently, it follows
from \cite[Chapter II, Theorem 1.10]{EngelNagel} that
\[\|\mf R(z,H_n)-\mf R(z,H)\|_{\mr{op}}=\left\|\int_0^\infty\mr e^{tz}(\mr e^{-t G_n}-K_t)\d t\right\|_{\mr{op}}
\leq\int_0^\infty\mr e^{t z}\|\mr e^{-t G_n}-K_t\|_{\mr{op}}\d t,\]
where the last inequality follows from \cite[Chapter II, Theorem 4 (ii)]{DU77}.
Given that
\[\int_0^\infty\mr e^{t z}\|\mr e^{-t G_n}-K_t\|_{\mr{op}}\d t\leq\int_0^\infty\mr e^{t z}\big(\|\mr e^{-t G_n}\|_{\mr{op}}+\|K_t\|_{\mr{op}}\big)\d t
\leq2\int_0^\infty\mr e^{t(z-\om)}\d t<\infty\]
whenever $\Re(z)<\om$, we get that \eqref{Equation: Resolvent Convergence Assumption 2} is
a consequence of \eqref{Equation: Semigroup Convergence Assumption 2} by an application of the dominated convergence theorem.

Let us then prove \eqref{Equation: Semigroup Convergence Assumption 2}.
Since the Hilbert-Schmidt norm dominates the operator norm, it suffices to prove that
\begin{align}
\label{Equation: H-S Convergence}
\sum_{u,v\in\ms V}\big(\mr e^{-tG_n}(u,v)-K_t(u,v)\big)^2=\sum_{u,v\in\ms V}\mbf E^u\left[\mr e^{-\langle L_t,V+\xi\rangle}\mbf 1_{\{X(t)=v\}}\mbf 1_{\{\tau_n\leq t\}}\right]^2
\end{align}
vanishes as $n\to\infty$ for all $t>0$ almost surely. By H\"older's inequality,
the right-hand side of \eqref{Equation: H-S Convergence} is bounded above by
\[\sum_{u,v\in\ms V}\mbf E^u\left[\mr e^{-2\langle L_t,V+\xi\rangle}\mbf 1_{\{X(t)=v\}}\right]\mbf P^u[\tau_n\leq t].\]
By mimicking our proof that $K_t$ is trace class, we know that
\[\sum_{u,v\in\ms V}\mbf E^u\left[\mr e^{-2\langle L_t,V+\xi\rangle}\mbf 1_{\{X(t)=v\}}\right]<\infty\]
for every $t>0$ almost surely. Thus, by dominated convergence, it suffices to prove that
\[\lim_{n\to\infty}\mbf P^u[\tau_n\leq t]=0\]
for every $u\in\ms V$ and $t>0$. Noting that
\[\mbf P^u\left[\max_{0\leq s\leq t} \msf d\big(0,X(s)\big)> n\right]
\leq\mbf P^u\left[\max_{0\leq s\leq t} \msf d\big(u,X(s)\big)> n-\msf d(0,u)\right]\]
for all $n\in\mbb N$ by the triangle inequality,
this follows directly from
the tail bound \eqref{Equation: Tail Bound}.
\end{proof}

\subsection{Step 5. Rigidity}

It now only remains to prove that the point process
\eqref{Equation: Eigenvalue Point Process}
is number rigid in the sense of Definition \ref{Definition: Rigidity}. The proof of
this amounts to a minor modification of the argument in \cite[Theorem 6.1]{GP17} (see also
\cite[Proposition 2.2]{GGL20}).

Let $B\subset\mbb C$ be a Borel set such that
$B\subset(-\infty,\de]+\mr i[-\tilde \de,\tilde \de]$
for some $\de,\tilde\de>0$.
Thanks to the trace identity \eqref{Equation: Trace Identity},
almost surely,
we can write 
\[\mc X_H(B)=\sum_{\la\in\si(H)\cap B}m_a(\la,H)\]
as the sum of the following three terms:
\begin{align}
\label{Equation: Rigidity 1}
&\sum_{\la\in\si(H)}m_a(\la,H)\,\mr e^{-t\la}-\mbf E\left[\sum_{\la\in\si(H)}m_a(\la,H)\,\mr e^{-t\la}\right]=\mr{Tr}[K_t]-\mbf E\big[\mr{Tr}[K_t]\big],\\
\label{Equation: Rigidity 2}
&\sum_{\la\in\si(H)\cap B}m_a(\la,H)\left(1-\mr e^{-t\la}\right),\\
\label{Equation: Rigidity 3}
&\mbf E\left[\sum_{\la\in\si(H)}m_a(\la,H)\,\mr e^{-t\la}\right]
-\sum_{\la\in\si(H)\setminus B}m_a(\la,H)\,\mr e^{-t\la}.
\end{align}
Since we choose the exponent $\al$ in the same way as
Theorem \ref{Theorem: Upper}, \eqref{Equation: Rigidity 1} converges to zero as $t\to0$
almost surely along a subsequence.
Next, we have that \eqref{Equation: Rigidity 2} is bounded above
in absolute value by
\[\mc X_H(B)\sup_{\ze\in[\om,\de]+\mr i[\al,\be]}|1-\mr e^{-t\ze}|,\]
where we recall that $\om$ is the random lower bound on the real part
of the points in $\mc X_H$.
Since $\mc X_H$ is real-bounded below and $B\subset(-\infty,\de]+\mr i[-\tilde \de,\tilde \de]$,
$\mc X_H(B)<\infty$ almost surely. Thus, \eqref{Equation: Rigidity 2}
converges to zero almost surely as $t\to0$.
Thus, $\mc X_H(B)$ is the almost sure limit of \eqref{Equation: Rigidity 3}
as $t\to0$, along a subsequence.
Given that \eqref{Equation: Rigidity 3}
is measurable with respect
to the configuration of points outside of $B$ for every $t$ and that the almost-sure limit of measurable functions is measurable (assuming the sigma algebra is complete),
we conclude that $\mc X_H(B)$ is measurable with respect
to the configuration outside of $B$. This then concludes the proof
of number rigidity, and thus of Theorem \ref{Theorem: Rigidity}.

\begin{remark}
\label{Remark: Mechanism}
Referring back to the point raised in Section \ref{Section: Mechanism},
we see that the function denoted $\mc N_B$ therein satisfies the relation
\begin{align}
\label{Equation: Understanding of NB}
\mc N_B\big(\si(H)\setminus B\big)=\lim_{n\to\infty}\left(\mbf E\left[\sum_{\la\in\si(H)}m_a(\la,H)\,\mr e^{-t_n\la}\right]
-\sum_{\la\in\si(H)\setminus B}m_a(\la,H)\,\mr e^{-t_n\la}\right)
\end{align}
with probability one, where $(t_n)_{n\in\mbb N}$ is a sparse enough
sequence that vanishes in the large $n$ limit. In particular, understanding the
precise form of $\mc N_B$ relies, among other things, on understanding
how the divergences of the two terms inside the limit on the right-hand side of
\eqref{Equation: Understanding of NB} somehow cancel out as $n\to\infty$.
\end{remark}

\section{Proof of Theorem \ref{Theorem: Lower}}
\label{sec: Proof of Lower}

\subsection{Step 1. General Lower Bound}

We begin by providing a lower bound for $\mbf{Var}\big[\mr{Tr}[K_t]\big]$
in the general setting of the statement of
Theorem \ref{Theorem: Lower}. This bound will then be shown to remain positive
as $t\to0$ in the cases labelled (1)--(3).

Recalling that $\ga$ is the positive definite covariance function of $\xi$,
if we denote the semi-inner-product
\[\langle f,g\rangle_\ga:=\sum_{u,v\in\mbb Z^d}f(u)\ga(u-v)g(v),\qquad f,g:\mbb Z^d\to\mbb R,\]
then our assumption that $\ga$ is nonnegative implies that $\langle f,g\rangle_\ga\geq0$
whenever $f$ and $g$ are nonnegative. In particular, we have that
\begin{align}
\label{Equation: Lower Bound - Covariance}
\mbf{Cov}_\xi\big[\mr e^{-\langle L^u_t,\xi\rangle},\mr e^{-\langle\tilde L^v_t,\xi\rangle}\big]
=\mr e^{\frac12\langle L_t^u,L_t^u\rangle_\ga+\frac12\langle\tilde L_t^v,\tilde L_t^v\rangle_\ga}\left(\mr e^{\langle L_t^u,\tilde L_t^v\rangle_\ga}-1\right)\geq0.
\end{align}
For every $u,v\in\mbb Z^d$ and $t>0$, denote the event
$J_t(u,v):=\{L^u_t=t\mbf 1_u\text{ and }\tilde L^v_t=t\mbf 1_v\}$.
Clearly, $J_t(u,v)\subset\{X^u(t)=u,\tilde X^v(t)=v\}$, and
by independence of $X^u$ and $\tilde X^v$,
\begin{align}
\label{Equation: Lower Bound - J_t Event}
\inf_{u,v\in\mbb Z^d}\mbf P[J_t(u,v)]=\inf_{v\in\mbb Z^d}\mbf P^v[X(s)=v\text{ for every }s\leq t]^2
\geq\mr e^{-2t}.
\end{align}

We now combine \eqref{Equation: Lower Bound - Covariance}
and \eqref{Equation: Lower Bound - J_t Event} to lower bound the variance of $\mr{Tr}[K_t]$:
By Proposition \ref{Proposition: Variance Formula}, we may write 
\begin{align}
\mbf{Var}\big[\mr{Tr}[K_t]\big]&\geq\sum_{u,v\in\mbb Z^d}\mbf E\Big[\mr e^{-\langle L^u_t+\tilde L^v_t,V\rangle}
\mr e^{\frac12\langle L_t^u,L_t^u\rangle_\ga+\frac12\langle\tilde L_t^v,\tilde L_t^v\rangle_\ga}\left(\mr e^{\langle L_t^u,\tilde L_t^v\rangle_\ga}-1\right)
\mbf 1_{J_t(u,v)}\Big]\nonumber\\
&=\sum_{u,v\in\mbb Z^d}\mr e^{-tV(u)-tV(v)}\mr e^{t^2\ga(0)}\left(\mr e^{t^2\ga(u-v)}-1\right)\mbf P[J_t(u,v)]\nonumber\\
&\geq\mr e^{-2t+t^2\ga(0)}\sum_{u,v\in\mbb Z^d}\mr e^{-tV(u)-tV(v)}\left(\mr e^{t^2\ga(u-v)}-1\right)\nonumber\\
&=\mr e^{-2t+t^2\ga(0)}\sum_{u,v\in\mbb Z^d}\mr e^{-t\msf d(0,u)^\de-t\msf d(0,v)^\de}\left(\mr e^{t^2\ga(u-v)}-1\right),\label{eq:VarLowLast}
\end{align}
where the first line comes from \eqref{Equation: Lower Bound - Covariance} and
the fact that $\mbf E[Y]\geq\mbf E[Y\mbf 1_E]$ for any nonnegative random variable $Y$
and event $E$, the second line comes from the definition of the event $J_t(u,v)$,
the third line comes from \eqref{Equation: Lower Bound - J_t Event},
and the last line comes from the assumption on $V$ stated in Theorem \ref{Theorem: Lower}.
As $\mr e^{-2t+t^2\ga(0)}\to1$ as $t\to0$, we obtain our general lower bound:
\begin{align}
\label{Equation: General Lower Bound}
\liminf_{t\to0}\mbf{Var}\big[\mr{Tr}[K_t]\big]\geq\liminf_{t\to0}\sum_{u,v\in\mbb Z^d}\mr e^{-t\msf d(0,u)^\de-t\msf d(0,v)^\de}\left(\mr e^{t^2\ga(u-v)}-1\right).
\end{align}
We now prove that the right-hand side of \eqref{Equation: General Lower Bound} is positive in cases (1)--(3).

\subsection{Step 2. Three Examples}

Suppose first that $\de\leq d/2$ and $\ga(v)=\mbf 1_{\{v=0\}}$.
On the integer lattice $\mbb Z^d$, it is easy to see that there exists a
constant $C>0$ such that $\msf c_n(0)\geq Cn^{d-1}$.
Therefore,
by an application of \eqref{Equation: General Lower Bound}, followed by
the inequality $\mr e^{x}-1\geq x$ for all $x\geq0$ and a Riemann sum, we have that
\begin{multline*}
\liminf_{t\to0}\mbf{Var}\big[\mr{Tr}[K_t]\big]\geq\liminf_{t\to0}\left(\mr e^{t^2}-1\right)\sum_{v\in\mbb Z^d}\mr e^{-2t\msf d(0,v)^\de}
\geq\liminf_{t\to0}t^2\sum_{n\in\mbb N\cup\{0\}}\msf c_n(0)\mr e^{-2tn^\de}\\
\geq C \liminf_{t\to0}t^{2-d/\de}t^{1/\de}\sum_{n\in t^{1/\de}\mbb N\cup\{0\}}n^{d-1}\mr e^{-2n}
\geq C \int_0^\infty x^{d-1}\mr e^{-2x}\d x>0.
\end{multline*}
Next, suppose that $\de\leq d-\be/2$ and that
$\ga(v)\geq\mc L\big(\msf d(0,v)+1\big)^{-\be}$ for some $0<\be<d$ and $\mc L>0$. Then,
\eqref{Equation: General Lower Bound}, the triangle inequality,
and the same arguments as in the previous case yield
\begin{align*}
&\liminf_{t\to0}\mbf{Var}\big[\mr{Tr}[K_t]\big]\\
&\geq\liminf_{t\to0}\sum_{u,v\in\mbb Z^d}\mr e^{-t\msf d(0,u)^\de-t\msf d(0,v)^\de}
\left(\mr e^{\mc Lt^2(\msf d(u,v)+1)^{-\be}}-1\right)\\
&\geq \mc L \liminf_{t\to0}t^{2}\sum_{u,v\in\mbb Z^d}\mr e^{-t\msf d(0,u)^\de-t\msf d(0,v)^\de}
\big(\msf d(0,u)+\msf d(0,v)+1\big)^{-\be}\\
&=\mc L \liminf_{t\to0}t^{2}\sum_{m,n\in\mbb N\cup\{0\}}\msf c_m(0)\msf c_n(0)\,\mr e^{-tm^\de-tn^\de}(m+n+1)^{-\be}\\
&\geq\mc L C^2\liminf_{t\to0}t^{2-2(d-1)/\de+\be/\de}\sum_{m,n\in t^{1/\de}\mbb N\cup\{0\}}(mn)^{d-1}\mr e^{-m^\de-n^\de}(m+n+t^\de)^{-\be}\\
&=\mc L C^2\liminf_{t\to0}t^{2-2(d-\be/2)/\de}\int_0^\infty\int_0^\infty\frac{(xy)^{d-1}}{(x+y)^\be}\mr e^{-x^\de-y^\de}\d x\dd y>0.
\end{align*}
Finally, suppose that $\de\leq d$ and $\inf_{v\in\mbb Z^d}\ga(v)>\mc L>0$. In this case we obtain that
\begin{multline*}
\liminf_{t\to0}\mbf{Var}\big[\mr{Tr}[K_t]\big]\geq\liminf_{t\to0}\left(\mr e^{\mc Lt^2}-1\right)\sum_{u,v\in\mbb Z^d}\mr e^{-t\msf d(0,u)^\de-t\msf d(0,v)^\de}\\
\geq\mc L C^2\liminf_{t\to0}t^2\left(\sum_{n\in\mbb N}n^{d-1}\mr e^{-2tn^\de}\right)^2
=\mc L C^2\liminf_{t\to0}t^{2-2d/\de}\left(\int_0^\infty x^{d-1}\mr e^{-2x}\d x\right)^2>0,
\end{multline*}
thus concluding the proof.

\bibliographystyle{plain}
\bibliography{Bibliography}
\end{document}